\documentclass{article}

\usepackage{amsmath}
\usepackage{amsthm}
\usepackage[T1]{fontenc}
\usepackage{amsfonts}
\usepackage{amssymb}
\usepackage{float}
\usepackage{url}
\usepackage{tikz}
\usepackage{hyperref}
\usetikzlibrary{automata}
\usetikzlibrary{shadows}
\usetikzlibrary{arrows}
\usetikzlibrary{shapes}
\usetikzlibrary{decorations.pathmorphing}
\usetikzlibrary{fit}
\usetikzlibrary{trees}
\usepackage{bbm}

\usepackage[latin1]{inputenc}
\usepackage{subfigure}

\usepackage[square,numbers]{natbib}

\tikzstyle{every picture}=[
  >=stealth', shorten >=1pt, node distance=1.44cm,auto,bend angle=45,initial text=,
  every state/.style={inner sep=0.75mm, minimum size=1mm},font=\scriptsize,
]

\newcommand{\Intervalle}[2]{\{#1,\ldots,#2\}}
\newcommand{\Card}[1]{\mathrm{Card}(#1)}
\newcommand{\indef}{\bot}

\newtheorem{lemma}{Lemma}
\newtheorem{corollary}{Corollary}

\newtheorem{proposition}{Proposition}
\newtheorem{definition}{Definition}
\newtheorem{example}{Example}

\begin{document} 

  \title{Derivatives of Approximate Regular Expressions}
  
  \author{J.-M.~Champarnaud 
  \and H.~Jeanne
  \and L.~Mignot
  }
  
  \date{}

  \maketitle
  \begin{abstract}
    Our aim is to construct a finite automaton recognizing the set of words that are at a bounded distance from some word of a given regular language. We define new regular operators, the similarity operators, based on a generalization of the notion of distance and we introduce the family of regular expressions extended to similarity operators, that we call AREs (Approximate Regular Expressions). We set formulae to compute the Brzozowski derivatives and the Antimirov derivatives of an ARE, which allows us to give a solution to the ARE membership problem and to provide the construction of two recognizers for the language denoted by an ARE.
As far as we know, the family of approximative regular expressions is introduced for the first time in this paper.
Classical approximate regular expression matching algorithms are approximate matching algorithms on regular expressions.
Our approach is rather to process an exact matching on approximate regular expressions.
  \end{abstract}

\section{Introduction}\label{se:int}

This paper addresses the problem of constructing a finite automaton that recognizes
the language of all the words that are at a distance less than or equal to a given positive integer $k$ 
from some word of a given regular language.
Our approach is based on the extension of regular expressions to approximate regular expressions (AREs)
that handle distance operators.
More precisely, we first define a new family of operators:
given an integer $k$, the $\mathbb{F}_k$ operator is such that,
for any regular language $L$, the language $\mathbb{F}_k(L)$ is the set of
all the words that are at a distance less than or equal to $k$ 
from some word of $L$. 
We then consider the family of approximate regular expressions
obtained from the family of regular expressions 
by adding the family of $\mathbb{F}_k$ operators to the set of regular operators.
We provide a formula that, 
given a regular language $L$, computes
the quotient of the language $\mathbb{F}_k(L)$ with respect to a symbol.
We finally extend the computation of Brzozowski derivatives~\cite{Brz64} (resp. of Antimirov derivatives~\cite{Ant96})
to the family of approximate regular expressions.
The first benefit of the derivation of an ARE is that it yields an elegant solution for the approximate membership problem.
%(without constructing an automaton).
Moreover, the set of Brzozowski derivatives (resp. of Antimirov derivatives) of an ARE is shown to be finite.
As a consequence, the derivation of an ARE enables the computation of a finite automaton that recognizes the language of this ARE.

The similarity between two words is generally measured by a distance and 
two basic types of distance called Hamming distance and Levenshtein distance (or edit distance) are generally considered.
%Our constructions provide a first kind of generalization since 
In our constructions
the similarity between two words
is handled by a word comparison function,
that is more general than a distance (for instance, a comparison function is not necessarily symmetrical).
It is the reason why we will speak of similarity operators rather than of distance operators.

%As a matter of fact,
The aim of this paper is to investigate the properties of the AREs family,
in particular to define formulae for computing the set of (Brzozowski or Antimirov) derivatives of an ARE
and to check the properties of this set. 
This theoretical study leads to a solution for the approximate membership problem
as well as to a solution for the approximate regular expression matching problem 
(based on the automaton associated with the set of derivatives of an ARE).
However, this paper is not an algorithmic contribution to the approximate regular expression matching problem:
it investigates new automaton-theoretic constructions
that hopefully make a sound foundation for the design of new approximate matching algorithms,
but it does not present new efficient algorithms.  

Let us recall that approximate matching consists in locating the segments of the text 
that approximately correspond to the pattern to be matched,
{\it i.e.\/} segments that do not present too many errors with respect to the pattern.
This research topic has numerous applications, in biology or in linguistics for example,
and many algorithms have been designed in this framework for more than thirty years
especially concerning approximate string matching
(see~\cite{CL06,Nav01} for a survey of such algorithms).
Two contexts can be distinguished: in the off-line case, that is when a pre-computing of the text is performed,
the basic tool is the construction of indexes~\cite{JTU96};
otherwise, 
the basic technique is dynamic programming~\cite{MM89}.
In both cases, automata constructions have been used,
either to represent an index~\cite{UW93,BYG96}
or to simulate dynamic programming~\cite{Hol02}.

Several studies address the problem of constructing a finite automaton that recognizes
the language of all the words that are at a distance less than or equal to a given positive integer $k$ 
from a given word. 
For instance this problem is considered in~\cite{Mab96}        
where Hamming distance is used and in~\cite{SM02}
where Levenshtein distance is used.
A challenging problem is to tackle the more general case
where the pattern is no longer a word but a regular expression~\cite{Nav07,WMM95}.
The solution
described in~\cite{Muz96}
first computes $k+1$ clones of some non-deterministic automaton
recognizing the language of the regular expression
and then interconnects these clones by a set of transitions that depends on the type of distance.   

As far as we know, the family of approximate regular expressions is introduced for the first time in this paper.
Approximate regular expression matching algorithms described in the papers above-cited
are approximate matching algorithms on regular expressions.
Our approach is rather to process an exact matching on approximate regular expressions.

This paper is an extended version of~\cite{CJM12}. Classical notions of language theory, such as derivative computation, are recalled in Section~\ref{se:pre}.
Section~\ref{se: dist cost comp} gives a formalization of the notion of word comparison function and provides a definition of the family of approximate regular expressions.
The usual case of Hamming and Levenshtein distances is addressed in Section~\ref{se:hamLevDeriv}.
Finally, Section~\ref{se: deriv are} is devoted to the general case and derivative-based constructions of an automaton from an approximate regular expression are described.

\section{Preliminaries}\label{se:pre}

  Given a set $X$, we denote by $\mathrm{Card}(X)$ the number of elements in $X$.

  A \emph{finite automaton} $A$ is a 5-tuple $(\Sigma,Q,I,F,\delta)$ with:
  \begin{itemize}
    \item $\Sigma$ the \emph{alphabet} (a finite set of symbols), 
    \item $Q$ a finite set of \emph{states},  
    \item $I \subset Q$ the set of \emph{initial states}, 
    \item $F \subset Q$ the set of \emph{final states},
    \item $\delta \subset Q\times\Sigma\times Q$ the set of \emph{transitions}.
  \end{itemize}
  The set $\delta$ is equivalent to the function from $Q\times\Sigma$ to $2^Q$ defined by: $q'\in\delta(q,a)$ if and only if $(q,a,q')\in\delta$. The domain of the function $\delta$ is extended to $2^Q\times \Sigma^*$ as follows: $\forall P\subset Q$, $\forall a\in\Sigma$, $\forall w\in\Sigma^*$, $\delta(P,\varepsilon) =P$, $\delta(P,a)=\bigcup_{p\in P}\delta(p,a)$ and $\delta(P,a\cdot w)=\delta(\delta(P,a),w)$. The automaton $A$ \emph{recognizes} the language $L(A)=\{w\in\Sigma^*\mid \delta(I,w)\cap F\neq\emptyset\}$. The automaton $A$ is \emph{deterministic} if $\Card{I}=1$ 
  and
  $\forall (q,a)\in Q\times \Sigma$, $\Card{\delta(q,a)}\leq 1$.
  
A \emph{regular expression} $E$ over an alphabet $\Sigma$ is inductively defined by:

 \centerline{$E=\emptyset$, $E=\varepsilon$, $E=a$,}
 
 \centerline{ $E=(F+G)$, $E=(F\cdot G)$, $E=(F^*)$}
 
 where $a$ is any symbol in $\Sigma$ and $F$ and $G$ are any two regular expressions.
 
 The \emph{language} $L(E)$ \emph{denoted by} $E$ is inductively defined by:
 
 \centerline{ $L(\emptyset)=\emptyset$, $L(a)=\{a\}$, $L(\varepsilon)=\{\varepsilon\}$,}
 
 \centerline{ $L(E+F)=L(E)\cup L(F)$, $L(E\cdot F)=L(E)\cdot L(F)$ and $L(F^*)=(L(F))^*$}
 
  where $a$ is any symbol in $\Sigma$, $F$ and $G$ are any two regular expressions, and for any $L_1,L_2\subset \Sigma^*$, 
  
  \centerline{$L_1\cup L_2=\{w\mid w\in L_1\vee w\in L_2\}$,}
  
  \centerline{$L_1\cdot L_2=\{w_1w_2\mid w_1\in L_1\wedge w_2\in L_2\}$}
  
  \centerline{and $L_1^*=\{w_1\cdots w_k\mid k\geq 1\wedge \forall j\in\{1,\ldots,k\},\ w_j\in L_1\}\cup\{\varepsilon\}$.}
  
A language $L$ is \emph{regular} if there exists a regular expression $E$ such that $L(E)=L$.
It has been proved by Kleene~\cite{Kle56} that a language is regular if and only if it is recognized by a finite automaton.

Given a language $L$ over an alphabet $\Sigma$ and a word $w$ in $\Sigma^*$, the \emph{membership problem} is to determine whether $w$ belongs to $L$. It can be solved by the computation of the boolean $\mathrm{r}(w,L)$ defined by:

	\centerline{
	    $\mathrm{r}(w,L)=
	    \left\{
	      \begin{array}{l@{\ }l}
	        1 & \text{ if } w\in L,\\
	        0 & \text{ otherwise.}
	      \end{array}
	    \right.$
	}

   The \emph{quotient of} $L$ w.r.t. a symbol $a$ is the language $a^{-1}(L)=\{w\in\Sigma^*\mid aw\in L\}$. It can be recursively computed as follows:

  \centerline{ \hfill$a^{-1}(\emptyset)=a^{-1}(\{\varepsilon\}) = a^{-1}(\{b\})=\emptyset$,\hfill $a^{-1}(\{a\})=\{\varepsilon\}$\hfill
   } 
  
  \centerline{\hfill$a^{-1}(L_1\cup L_2) =a^{-1}(L_1)\cup a^{-1}(L_2)$,\hfill  $a^{-1}(L_1^*) =a^{-1}(L_1)\cdot L_1^*$\hfill}

  \centerline{$a^{-1}(L_1\cdot L_2)=
    \left\{
      \begin{array}{l@{\ }l}
        a^{-1}(L_1)\cdot L_2 \cup a^{-1}(L_2) & \text{ if }\mathrm{r}(\varepsilon,L_1)=1,\\
	a^{-1}(L_1)\cdot L_2 & \text{ otherwise.}\\
      \end{array}
    \right.$
  }
  
\noindent    The quotient $w^{-1}(L)$ of $L$ w.r.t. a word $w$ in $\Sigma^*$ is the set $\{w'\in\Sigma^*\mid w\cdot w'\in L\}$. It can be recursively computed as follows: $\varepsilon^{-1}(L)=L$, $(aw')^{-1}(L)=w'^{-1}(a^{-1}(L))$ with $a\in\Sigma$ and $w'\in\Sigma^+$. 
The Myhill-Nerode Theorem~\cite{Myh57,Ner58}
 states that a language $L$ is regular if and only if the set of quotients $\{u^{-1}(L)\mid u\in\Sigma^*\}$ is finite. 
 
Since $\mathrm{r}(w,L)=\mathrm{r}(\varepsilon,w^{-1}(L))$, the membership problem can be solved using the quotient formulae and the following straightforward computation of $\mathrm{r}(\varepsilon,L)$:

\centerline{
  $\mathrm{r}(\varepsilon,\{a\})=\mathrm{r}(\varepsilon,\emptyset)= 0$, $\mathrm{r}(\varepsilon,\{\varepsilon\})= 1$,
}

\centerline{
  $\mathrm{r}(\varepsilon,L_1\cup L_2)=\mathrm{r}(\varepsilon,L_1)\vee \mathrm{r}(\varepsilon,L_2)$, $\mathrm{r}(\varepsilon,L_1\cdot L_2)=\mathrm{r}(\varepsilon,L_1)\wedge \mathrm{r}(\varepsilon,L_2)$,
}

\centerline{
  $\mathrm{r}(\varepsilon,L_1^*)=1$.
}
  
The notion of derivative of an expression has been introduced by Brzozowski~\cite{Brz64}. The 
 derivative
  of an expression $E$ w.r.t. a word $w$ 
  is
   an expression denoting the quotient of $L(E)$ w.r.t. $w$.
    Let $E$ be a regular expression over an alphabet $\Sigma$ and let $a$ and $b$ be two distinct symbols of $\Sigma$. The \emph{derivative of} $E$ w.r.t. $a$ is the expression $\frac{d}{d_a}(E)$ inductively computed as follows:

      \centerline{
	$\frac{d}{d_a}(\emptyset)=\frac{d}{d_a}(\varepsilon)=\frac{d}{d_a}(b)=\emptyset$,\ \ \ \  $\frac{d}{d_a}(a)=\varepsilon$,
      }
      
      \centerline{
	$\frac{d}{d_a}(F^*)=\frac{d}{d_a}(F)\cdot F^*$,\ \ \ \   $\frac{d}{d_a}(F+G)=\frac{d}{d_a}(F)+\frac{d}{d_a}(G)$
      }
      
      \centerline{
	$\frac{d}{d_a}(F\cdot G)=
	  \left\{
	    \begin{array}{l@{ }l}
	      \frac{d}{d_a}(F)\cdot G+\frac{d}{d_a}(G)&\text{ if } \mathrm{r}(\varepsilon,L(F))=1,\\
	      \frac{d}{d_a}(F)\cdot G & \text{ otherwise.}
	    \end{array}
	  \right.
	$
      }

 \noindent The derivative of $E$ is extended to words of $\Sigma^*$ as follows:

  \centerline{$\frac{d}{d_{\varepsilon}}(E)=E$, $\frac{d}{d_{aw}}(E)=\frac{d}{d_{w}}(\frac{d}{d_{a}}(E))$.}
  
   Since $w^{-1}(L(E))=L(\frac{d}{d_w}(E))$, it holds $\mathrm{r}(w,L(E))=\mathrm{r}(\varepsilon,L(\frac{d}{d_w}(E)))$. For convenience, we set $\mathrm{r}(w,E)=\mathrm{r}(w,L(E))$. Notice that the boolean $\mathrm{r}(\varepsilon,E)$ can be inductively computed as follows:   

\centerline{
  $\mathrm{r}(\varepsilon,a)=\mathrm{r}(\varepsilon,\emptyset)= 0$, $\mathrm{r}(\varepsilon,\varepsilon)= 1$,
}

\centerline{
  $\mathrm{r}(\varepsilon,E_1\cup E_2)=\mathrm{r}(\varepsilon,E_1)\vee \mathrm{r}(\varepsilon,E_2)$, $\mathrm{r}(\varepsilon,E_1\cdot E_2)=\mathrm{r}(\varepsilon,E_1)\wedge \mathrm{r}(\varepsilon,E_2)$,
}

\centerline{
  $\mathrm{r}(\varepsilon,E_1^*)=1$.
}

As a consequence, derivation provides a syntactical solution for the membership problem.

 Notice that the set $\mathcal{D}_E$ of derivatives of an expression $E$ is not necessarily finite.
 It has been proved by Brzozowski~\cite{Brz64} that it is sufficient to 
use the ACI equivalence
(that is based on the associativity, the commutativity and the idempotence of the sum of expressions)
to obtain a finite set of derivatives:
the set $\mathcal{D}'_E$ of \emph{dissimilar derivatives}.
Given a class of ACI-equivalent expressions,
a unique representative
can be obtained after deleting parenthesis (associativity),
ordering terms of each sum (commutativity) and deleting redundant subexpressions (idempotence).
Let $E_{\sim_s}$ be the unique representative of the class of the expression $E$. The set of dissimilar derivatives can be computed as follows:

      \centerline{
	$\frac{d'}{d'_a}(\emptyset)=\frac{d'}{d'_a}(\varepsilon)=\frac{d'}{d'_a}(b)=\emptyset$, $\frac{d'}{d'_a}(a)=\varepsilon$,
      }

      \centerline{
	$\frac{d'}{d'_a}(E+F)=(\frac{d'}{d'_a}(F)+\frac{d'}{d'_a}(G))_{\sim_s}$, $\frac{d'}{d'_a}(F^*)=\frac{d'}{d'_a}(F)\cdot F^*$,
      }
      
      \centerline{
	$\frac{d'}{d'_a}(F\cdot G)=
	  \left\{
	    \begin{array}{l@{ }l}
	      (\frac{d'}{d'_a}(F)\cdot G+\frac{d'}{d'_a}(G))_{\sim_s}&\text{ if } \mathrm{r}(\varepsilon,F)=1,\\
	      (\frac{d'}{d'_a}(F)\cdot G)_{\sim_s} & \text{ otherwise.}
	    \end{array}
	  \right.
	$
      }      

  \noindent The \emph{dissimilar derivative finite automaton} $B'(E)=(\Sigma,Q,\{q_0\},F,\delta)$ of a
  regular expression $E$ over an alphabet $\Sigma$
is defined by:
\begin{itemize}
  \item $Q =\mathcal{D}'_{E}$, 
  \item $q_0 =E$, 
  \item $F =\{q\in Q\mid \varepsilon\in L(q)\}$, 
  \item $\delta =\{(q,a,q')\in Q\times\Sigma\times Q\mid \frac{d'}{d'_{a}}(q)=q'\}$.
\end{itemize}
 The automaton $B'(E)$ is deterministic and it recognizes the language $L(E)$.
 Its size can be exponentially larger than the number of symbols of $E$.
 
  Antimirov's
  algorithm~\cite{Ant96} constructs a finite automaton from a regular expression $E$. It is  based on the \emph{partial derivative} computation.
    The partial derivative of a regular expression $E$ w.r.t. a symbol $a$ is the set $\frac{\partial}{\partial_a}(E)$ of expressions defined as follows:
    
      \centerline{
	  $\frac{\partial}{\partial_a}(\emptyset)=\frac{\partial}{\partial_a}(\varepsilon)=\frac{\partial}{\partial_a}(b)=\emptyset$, $\frac{\partial}{\partial_a}(a)=\{\varepsilon\}$,
      }
      
      \centerline{
	  $\frac{\partial}{\partial_a}(F+G)=\frac{\partial}{\partial_a}(F) \cup \frac{\partial}{\partial_a}(G)$, $\frac{\partial}{\partial_a}(F^*)=\frac{\partial}{\partial_a}(F)\cdot F^*$,
      }
      
      \centerline{
	  $\frac{\partial}{\partial_a}(F\cdot G)=
	    \left\{
	      \begin{array}{c@{\ }l}
		\frac{\partial}{\partial_a}(F)\cdot G \cup \frac{\partial}{\partial_a}(G)&\text{ if } \mathrm{r}(\varepsilon,F)=1,\\
		\frac{\partial}{\partial_a}(F)\cdot G & \text{ otherwise,}
	      \end{array}
	    \right.
	  $\
      }
      
  \noindent with for any set $\mathcal{E}$ of expressions, $\mathcal{E}\cdot F=\bigcup_{E\in\mathcal{E}} E\cdot F$.
  
  \noindent The partial derivative of $E$ is extended to words of $\Sigma^*$ as follows:
  
  \centerline{$\frac{\partial}{\partial_{\varepsilon}}(E)=\{E\}$, $\frac{\partial}{\partial_{aw}}(E)=\frac{\partial}{\partial_{w}}(\frac{\partial}{\partial_{a}}(E))$,}
    
  \noindent with for a set $\mathcal{E}$ of expressions, $\frac{\partial}{\partial_a}(\mathcal{E})=\bigcup_{E\in\mathcal{E}}\frac{\partial}{\partial_a}(E)$.
  \noindent Every element of the partial
derivative of $E$ w.r.t. a word $w$ in $\Sigma^*$ is called a \emph{derivated term
of} $E$ w.r.t. $w$. The \emph{set of the derivated terms of} $E$ is the union of the
sets of the derivated terms of $E$ w.r.t. $w$, for all $w$ in $\Sigma^*$. Antimirov~\cite{Ant96} has shown that 
the set $\mathcal{DT}_E$ of the derivated terms of $E$
is such that $\Card{\mathcal{DT}_E}\leq n+1$, where $n$ is the number of symbols of $E$.

Furthermore, for any word $w$ in $\Sigma^*$, $\bigcup_{E'\in\frac{\partial}{\partial_w}(E)} L(E')=w^{-1}(L(E))$. Consequently, the partial derivation provides another syntactical solution for the membership problem as well as  a finite automaton computation. Indeed, it can be shown that $\mathrm{r}(w,E)=\bigvee_{E'\in\frac{\partial}{\partial_w}(E)} \mathrm{r}(\varepsilon,E')$.

The \emph{derivated term finite automaton} $A(E)=(\Sigma,Q,\{q_0\},F,\delta)$ of a  regular expression $E$
is defined as follows: 
\begin{itemize}
  \item $Q =\mathcal{DT}_E$, 
  \item $q_0 =E$, 
  \item $F =\{q\in Q\mid \mathrm{r}(\varepsilon,q)=1\}$, 
  \item $\delta =\{(q,a,q')\in Q\times\Sigma\times Q\mid q'\in \frac{\partial}{\partial_a}(q)\}$.
\end{itemize}
 The automaton $A(E)$ recognizes the language $L(E)$.
 
In this paper, we consider the \emph{approximate membership problem} that is defined as follows:

Given a regular expression $E$ over an alphabet $\Sigma$, a word $w$ in $\Sigma^*$, a function $\mathbb{F}$ from $\Sigma^*\times \Sigma^*$ to $\mathbb{N}$ and an integer $k$, is there a word $w'$ in $L(E)$ satisfying $\mathbb{F}(w,w')\leq k$ ?

In the following, we provide a syntactical solution for the approximate membership problem in the case where the function $\mathbb{F}$ satisfies specific properties.
 
\section{Comparison Functions: Symbols, Sequences and Words}\label{se: dist cost comp}
  
Let $\Sigma$ be an alphabet, $S=\Sigma\ \cup\ \{\varepsilon\}$ and $X$ be a subset of $S\times S$.
A \emph{cost function} $\mathrm{C}$ over $X$ is a function from $X$ to $\mathbb{N}$
satisfying 
\textbf{Condition 1:}
 for all $\alpha$ in $S$, $\mathrm{C}(\alpha,\alpha)=0$.
For any 
pair
 $(\alpha,\beta)$ in $S\times S$ such that $\mathrm{C}(\alpha,\beta)$ is not defined,
let us set $\mathrm{C}(\alpha,\beta)=\indef$.
Consequently, a cost function can be viewed as a function from $S\times S$ to $\mathbb{N}\cup\{\indef\}$
satisfying 
 Condition $1$.
 Since we use 
 $\indef$ 
 to deal with undefined computation, we 
 set 
  for all $x$ in $\mathbb{N}\cup\{\indef\}$, $\indef+x=x+\indef=x-\indef=\indef-x=\indef$ and for all integers $x,y$ in $\mathbb{N}$, $x-y=\indef$ when $y>x$. 
   A cost function can be represented by a 
directed and labelled graph $\mathrm{C}=\{S,V\}$ where $V$ is a subset of 
$S\times (\mathbb{N} \cup \{\indef\})\times S$
 such that for all $(\alpha,\beta)$ in $S\times S$, $\mathrm{C}(\alpha,\beta)=k \Leftrightarrow (\alpha,k,\beta)\in V$.
 Transitions labelled by $\indef$ can be omitted in the graphical representation, 
as well as the implicit transitions $(\alpha,0,\alpha)$ (See Example~\ref{ex cost func as graph}).
  
  \begin{example}\label{ex cost func as graph}
    Let $\Sigma=\{a,b,c\}$. Let $\mathrm{C}$ be the cost function defined as follows:
    
    \begin{minipage}{0.45\linewidth}
    \centerline{
      $\mathrm{C}(x,y)=
        \left\{
          \begin{array}{l@{\ }l}
            0 & \text{ if } x=y,\\
            4 & \text{ if } x=a \wedge y=c,\\
            3 & \text{ if } x=c \wedge y=a,\\
            1 & \text{ if } x\in\{a,c\} \wedge y=b,\\
            \indef & \text{ otherwise.}\\
          \end{array}
        \right.
      $
    }
    \end{minipage}
    \hfill
    \begin{minipage}{0.5\linewidth}  
  \begin{figure}[H]
    \centerline{ 
      \begin{tikzpicture}[node distance=2cm,bend angle=30]   
	    \node[state] (b) {$b$};
	    \node[state] (a) [left of=b]{$a$};
	    \node[state] (c) [right of=b] {$c$};  	    	 
	    \node[state] (eps) [right of=c] {$\varepsilon$};	    	  
	    \path[->]
	      (a)   edge [bend left] node {4} (c)
	      (a)   edge node {1} (b)
	      (c)   edge [bend left] node {3} (a)
	      (c)   edge [swap] node {1} (b);	            
      \end{tikzpicture}
    }   
    \caption{The cost function $\mathrm{C}$}
    \label{fig ex cost function graph}
  \end{figure}
  \end{minipage}
    
    \vspace{0.5\baselineskip}
    The cost function $\mathrm{C}$ can be represented by the graph in Figure~\ref{fig ex cost function graph}.
  \end{example}
  
Given a positive integer $k$ we now consider the set $S^k$ of all 
the {\emph sequences} $s=(s_1,\ldots,s_k)$ of size $k$ made of elements of $S$.
A \emph{sequence comparison function} is a function $\mathcal F$ from $\bigcup_{k\in\mathbb{N}}S^k\times S^k$ 
to $\mathbb{N}\cup\{\indef\}$.
Given a 
pair
$(s,s')$ of sequences with the same size, ${\mathcal F}(s,s')$ either is an integer or is undefined.
In the following we will consider sequence comparison functions $\mathcal F$ satisfying
\textbf{Condition 2:} 
$\mathcal F$ is defined from a given cost function $\mathrm{C}$ over $S\times S$,
and 
\textbf{Condition 3:} 
$\mathcal F$ is a \emph{symbol-wise} comparison function,
that is, 
for any two sequences $s=(s_1,\ldots,s_n)$ and $s'=(s'_1,\ldots,s'_n)$, it holds:

   \centerline{
    \begin{tabular}{l@{\ }l@{\ }l}
      $\mathcal{F}(s,s')$ & $=\mathcal{F}((s_1),(s'_1))+\mathcal{F}((s_2,\ldots,s_n),(s'_2,\ldots,s'_n))$ & $=\sum_{k\in\Intervalle{1}{n}}\mathcal{F}((s_k),(s'_k)).$\\
    \end{tabular}
   }

\noindent We consider that those functions satisfy Condition~1
, \emph{i.e.}
  for all $\alpha$ in $S$, $\mathcal{F}((\alpha),(\alpha))=0$.
Consequently, for any 
pair
of sequences
 $s=(s_1,\ldots,s_k)$ and $s'=(s'_1,\ldots,s'_k)$ such that $k>1$,
\textbf{Condition 4}
 is satisfied:
 if there exists an integer $k'$ in $\Intervalle{1}{k}$ such that $s_{k'}=s'_{k'}=\varepsilon$, then:
  
  \centerline{ $\mathcal{F}(s,s')=
    \left\{
      \begin{array}{l@{\ }l}
        \mathcal{F}((s_2,\ldots,s_k),(s'_2,\ldots,s'_k)) & \text{ if }k'=1,\\ 
        \mathcal{F}((s_1,\ldots,s_{k-1}),(s'_1,\ldots,s'_{k-1})) & \text{ if }k'=k,\\ 
        \mathcal{F}((s_1,\ldots,s_{k'-1},s_{k'+1},\ldots,s_k),(s'_1,\ldots,s'_{k'-1},s'_{k'+1},\ldots,s'_k)) & \text{ otherwise.}\\
      \end{array}
    \right.$}
  
\noindent As a consequence of Condition $3$, a symbol-wise sequence comparison function is defined by
the images of the 
pairs
of sequences of size $1$.
Notice that a sequence comparison function is not necessarily symbol-wise, \emph{e.g.}  for a given cost function $\mathrm{F}$, $\mathcal{F}((s_1,\ldots,s_n),(s'_1,\ldots,s'_n))=\sum_{k\in\Intervalle{1}{n}} \mathrm{F}(s_k,s'_k)^k$.

 \noindent Two of the most well-known symbol-wise sequence comparison functions are the Hamming one ($\mathcal{H}$) and the Levenshtein one ($\mathcal{L}$) respectively defined for any integer $n>0$ and for any 
pair
of sequences $s=(s_1,\ldots,s_n)$ and $s'=(s'_1,\ldots,s'_n)$ in $S^n\times S^n$ by:
  
  \centerline{
    $\mathcal{H}(s,s')=\sum_{k\in\Intervalle{1}{n}} \mathrm{H}(s_k,s'_k)$,\ \ \ 
    $\mathcal{L}(s,s')=\sum_{k\in\Intervalle{1}{n}} \mathrm{L}(s_k,s'_k)$,
  }
  
  \noindent with $\mathrm{H}$ and $\mathrm{L}$ the two cost functions respectively defined for all $a,b$ in $\Sigma\cup\{\varepsilon\}$ by:

    \centerline{
      $\mathrm{H}(a,b)=
        \left\{
          \begin{array}{c@{\ }l}
            \indef & \text{ if } (a=\varepsilon\vee b=\varepsilon)\wedge (a,b)\neq(\varepsilon,\varepsilon),\\
            1 & \text{ if } a\neq b,\\
            0 & \text{ otherwise,}\\
          \end{array}
        \right.
      $ and 
      $\mathrm{L}(a,b)=
        \left\{
          \begin{array}{c@{\ }l}
            1 & \text{ if } a\neq b,\\
            0 & \text{ otherwise.}\\
          \end{array}
        \right.
      $
  }

Let us now explain how a word comparison function can be deduced from a sequence comparison function.
  Let $w$ be a word in $\Sigma^*$ and $|w|$ be its \emph{length}.
  The sequence $s=(s_1,\ldots,s_n)$ in $S^n$ is said to be a 
  \emph{split-up}
  of $w$ if $s_1\cdots s_n=w$. The integer $n$ is the \emph{size} of $s$.
 The set of all the 
 split-ups
 of size $k$ of a word $w$ is denoted by 
 $\mathrm{Split}_k(w)$
  and
 the set of all the 
 split-ups
  of $w$ is denoted by 
  $\mathrm{Split}(w)$.

    Let $\mathcal{F}$ be a sequence comparison function,
$(u,v)$ be a 
pair
 of words of $\Sigma^*$, and $k$ be a positive integer.
We consider the following sets:

\centerline{$ Y(u,v)=\{\mathcal{F}(u',v')\mid \exists k\in\mathbb{N}, k\geq 1 \wedge (u',v')\in \mathrm{Split}_k(u)\times \mathrm{Split}_k(v) \}\cap \mathbb{N},$}

\centerline{$Y_m(u,v)=\{\mathcal{F}(u',v')\mid \exists k\in\mathbb{N}, 1\leq k \leq m \wedge (u',v')\in \mathrm{Split}_k(u)\times \mathrm{Split}_k(v) \}\cap \mathbb{N}$.}

  \begin{definition}\label{def split up comp func}
    Let $\mathcal{F}$ be a sequence comparison function.
 The \emph{word comparison function} associated with $\mathcal{F}$ is the function $\mathbb{F}$ from $\Sigma^*\times \Sigma^*$ to $\mathbb{N}\cup\{\indef\}$ defined by:
 
  \centerline{\hfill$\mathbb{F}(u,v)=
        \mathrm{min}\{Y(u,v)\}\ \text{\ if\ } Y(u,v) \neq\emptyset,$\hfill
        $\mathbb{F}(u,v)=\indef  \text{\ otherwise.}$\hfill
  }
  \end{definition}
  
  Notice that a word comparison function is not necessarily symmetrical. Indeed, some problems can be modelized with a non-symmetrical function. For instance, given two words $w$ and $w'$, can $w$ be obtained from $w'$ by deleting some letters, \emph{i.e.} is $w$ a subword of $w'$? Such a problem can be modelized by the word comparison function $\mathbb{D}$ associated to the symbol-wise comparison function $\mathcal{D}$ defined for any pair of sequences of length $1$ by:
  
  \centerline{
  $\forall (\alpha,\beta)\in (\Sigma\cup\{\varepsilon\})^2$, $\mathcal{D}((\alpha),(\beta))=
    \left\{
      \begin{array}{l@{\ }l}
        0 & \text{ if }\alpha=\beta,\\
        1 & \text{ if }\alpha=\varepsilon\wedge \beta\in\Sigma,\\
        \indef & \text{ otherwise.} 
      \end{array}
    \right.
  $
  }
 
 It can be shown that for any two words $w$ and $w'$ in $\Sigma^*$: 
  
  \centerline{
  $\mathbb{D}(w,w')=
    \left\{
      \begin{array}{l@{\ }l}
        \indef & \text{ if }w\text{ is not a subword of }w',\\
        |w'|-|w| & \text{ otherwise.} 
      \end{array}
    \right.
  $
  }

  In the case of a sequence comparison function based on a cost function,
the whole set $\mathbb{N}$ needs not to be considered.
Indeed, according to Condition $4$,
if $u\neq\varepsilon$ or $v\neq\varepsilon$, 
then $Y(u,v)=Y_{|u|+|v|}(u,v)$ and we can write:

  \centerline{$\mathbb{F}(u,v)=
    \left\{
      \begin{array}{l@{\ }l}
        0 & \text{\ if\ } u=v=\varepsilon,\\
        \mathrm{min}\{Y_{|u|+|v|}(u,v)\} & \text{\ if\ } (u,v)\neq(\varepsilon,\varepsilon) \wedge Y_{|u|+|v|}(u,v)\neq\emptyset,\\
        \indef & \text{\ otherwise.}\\
      \end{array}
    \right.$
  }

  The \emph{Hamming distance} $\mathbb{H}$
and the \emph{Levenshtein distance} $\mathbb{L}$ are the word comparison functions respectively associated
 to the sequence comparison functions $\mathcal{H}$ and $\mathcal{L}$.
 Both of them satisfy the properties of word distances\footnote{A \emph{word distance} $\mathbb{D}$ is a word comparison function satisfying the three following properties for all $x,y,z\in \Sigma^*$: \textbf{(1)} $\mathbb{D}(x,y)=0 \Rightarrow x=y$, \textbf{(2)} $\mathbb{D}(x,y)=\mathbb{D}(y,x)$, \textbf{(3)} 
 $\mathbb{D}(x,y)+\mathbb{D}(y,z)\geq \mathbb{D}(x,z)$.}.  
Notice that in the following we will handle word comparison functions that are not necessarily distances (see Example~\ref{ex cost func as graph} for the definition of a nonsymmetrical cost function).
  
  \begin{example}\label{ex word comp}
    Let $\mathrm{C}$ be the cost function defined in Example~\ref{ex cost func as graph}. Let $s=(s_1)$ and $s'=(s'_1)$ be two sequences
of size 1. 
We define four symbol-wise sequence comparison functions
by setting the images of the 
pairs
 of sequences of size 1 from the cost function $\mathrm{C}$.
    
    \centerline{
      \begin{tabular}{l@{\  \ \ }l}
        $\rightarrow^{\mathrm{C}}(s,s')= \mathrm{C}(s_1,s'_1)$, & $\leftrightarrow^{\mathrm{C}}(s,s')=\mathrm{min}\{\mathrm{C}(s_1,s'_1),\mathrm{C}(s'_1,s_1)\}$,\\
        $\leftarrow^{\mathrm{C}}(s,s')=\mathrm{C}(s'_1,s_1)$, & $\rightrightarrows^{\mathrm{C}}(s,s')= \mathrm{min}_{x\in\Sigma\cup\{\varepsilon\}} \{\mathrm{C}(s_1,x)+\mathrm{C}(s'_1,x)\}$.\\
      \end{tabular}
    }
    
    Let us consider the two split-ups $s=(a,c,a)$ and $s'=(c,a,c)$. According to Figure~\ref{fig ex split up comp}, it holds:
    
    \begin{minipage}{0.30\linewidth}
    \centerline{
      \begin{tabular}{l}
        $\rightarrow^{\mathrm{C}}(s,s')=11$,\\
        $\leftarrow^{\mathrm{C}}(s,s')=10$,\\
        $\leftrightarrow^{\mathrm{C}}(s,s')=9$,\\
        $\rightrightarrows^{\mathrm{C}}(s,s')=6$.\\
      \end{tabular}
    }
    \end{minipage}
    \hfill
    \begin{minipage}{0.65\linewidth}
    \begin{figure}[H]
    \centerline{ 
      \begin{tabular}{l@{\ }|@{\ }l}
      \begin{tikzpicture}[node distance=0.6cm]  
	    \node (s1) {$s=($};  
	    \node (op) [below left of= s1] {$\rightarrow^{\mathrm{C}}:$}; 
	    \node (s11) [right of= s1]{$a$};
	    \node (s12) [right of= s11]{$c$};
	    \node (s13) [right of= s12] {$a$}; 
	    \node (s1f) [right of= s13] {$)$};  
	    \node (s2)  [below right of= op]{$s'=($};  
	    \node (s21) [right of= s2]{$c$};
	    \node (s22) [right of= s21]{$a$};
	    \node (s23) [right of= s22] {$c$}; 
	    \node (s2f) [right of= s23] {$)$};  	        	  
	    \path[->]
	      (s11)   edge  node {4} (s21)
	      (s12)   edge  node {3} (s22)
	      (s13)   edge  node {4} (s23)
	      (s1f)   edge [bend left]  node {11} (s2f);	            
      \end{tikzpicture}
      &
      \begin{tikzpicture}[node distance=0.6cm]   
	    \node (s1) {$s=($}; 
	    \node (op) [below left of= s1] {$\leftrightarrow^{\mathrm{C}}:$}; 
	    \node (s11) [right of= s1]{$a$};
	    \node (s12) [right of= s11]{$c$};
	    \node (s13) [right of= s12] {$a$}; 
	    \node (s1f) [right of= s13] {$)$};  
	    \node (s2)  [below right of= op]{$s'=($};  
	    \node (s21) [right of= s2]{$c$};
	    \node (s22) [right of= s21]{$a$};
	    \node (s23) [right of= s22] {$c$}; 
	    \node (s2f) [right of= s23] {$)$};  	        	  
	    \path[->]
	      (s21)   edge[swap]  node {3} (s11)
	      (s12)   edge  node {3} (s22)
	      (s23)   edge[swap]  node {3} (s13)
	      (s1f)   edge [bend left]  node {9} (s2f);	            
      \end{tikzpicture}
      \\
      \hline
      \begin{tikzpicture}[node distance=0.6cm]  
	    \node (s1) {$s=($};  
	    \node (op) [below left of= s1]{$\leftarrow^{\mathrm{C}}:$}; 
	    \node (s11) [right of= s1]{$a$};
	    \node (s12) [right of= s11]{$c$};
	    \node (s13) [right of= s12] {$a$}; 
	    \node (s1f) [right of= s13] {$)$};  
	    \node (s2)  [below right of= op]{$s'=($};  
	    \node (s21) [right of= s2]{$c$};
	    \node (s22) [right of= s21]{$a$};
	    \node (s23) [right of= s22] {$c$}; 
	    \node (s2f) [right of= s23] {$)$};  	        	  
	    \path[->]
	      (s21)   edge[swap]  node {3} (s11)
	      (s22)   edge[swap]  node {4} (s12)
	      (s23)   edge[swap]  node {3} (s13)
	      (s1f)   edge [bend left]  node {10} (s2f);	            
      \end{tikzpicture}
      &
      \begin{tikzpicture}[node distance=0.6cm]  
	    \node (s1) {$s=($};  
	    \node (s11) [right of= s1]{$a$};
	    \node (s12) [right of= s11]{$c$};
	    \node (s13) [right of= s12] {$a$}; 
	    \node (s1f) [right of= s13] {$)$};  
	    \node (b1) [below of= s11]{$b$};
	    \node (op)[left of= b1,node distance=1cm] {$\rightrightarrows^{\mathrm{C}}:$}; 
	    \node (b2) [below of= s12]{$b$};
	    \node (b3) [below of= s13] {$b$}; 
	    \node (s21) [below of= b1]{$c$};
	    \node (s2)  [left of= s21]{$s'=($};  
	    \node (s22) [right of= s21]{$a$};
	    \node (s23) [right of= s22] {$c$}; 
	    \node (s2f) [right of= s23] {$)$};  	        	  
	    \path[->]
	      (s11)   edge  node {1} (b1)
	      (s12)   edge  node {1} (b2)
	      (s13)   edge  node {1} (b3)
	      (s21)   edge[swap]  node {1} (b1)
	      (s22)   edge[swap]  node {1} (b2)
	      (s23)   edge[swap]  node {1} (b3)
	      (s1f)   edge [bend left]  node {6} (s2f);	            
      \end{tikzpicture}
      \\
    \end{tabular}      
    }  
    \caption{Examples of sequence comparisons}
    \label{fig ex split up comp}
  \end{figure}
    \end{minipage}
  \end{example}
  
  Any word comparison function can be used as a language operator in order to compute the set of words that are at a bounded distance from some word of a given language.
  
  \begin{definition}\label{def fkl}
    Let $L$ be a language over an alphabet $\Sigma$, $\mathbb{F}$ a word comparison function and $k$ an integer in $\mathbb{N}\cup\{\indef\}$. Then:

    \centerline{$\mathbb{F}_{k}(L)=
      \left\{
        \begin{array}{l@{\ }l}
          \{w\in\Sigma^*\mid\exists u\in L, \mathbb{F}(w,u)\in\Intervalle{0}{k}\} & \text{ if } k\in\mathbb{N},\\
          \emptyset & \text{ otherwise.}\\
        \end{array}
      \right.
    $
    }
  \end{definition}
    
  The operator $\mathbb{F}_{k}$ is called a \emph{similarity operator}. Let us notice that $\mathbb{F}_{k}(\mathbb{F}_{k'}(L))$ is not necessarily equal to $\mathbb{F}_{k+k'}(L)$. Indeed, let us consider the three languages $L_1=\mathbb{F}_{1}(\{a\})$, $L_2=\mathbb{F}_{1}(\mathbb{F}_{1}(\{a\}))$ and $L_3=\mathbb{F}_{2}(\{a\})$ over the alphabet $\Sigma=\{a,b\}$ with $\mathbb{F}$ the word comparison function associated with the symbol-wise sequence comparison function $\mathcal{F}$ defined for any symbol $\alpha,\beta$ by $\mathcal{F}((\alpha),(\beta))=0$ if $\alpha=\beta$, $\mathcal{F}((\alpha),(\beta))=2$ otherwise. Then $L_1=L_2=\{a\}$ whereas
 $L_3=\{\varepsilon,a,b,aa, ab, ba\}$.
 
 \begin{definition}      
  An \emph{approximate regular expression}\footnote{The fact that any ARE denotes a regular language is proved in Corollary~\ref{cor lang are rat}.} (\textbf{ARE}) $E$ over an alphabet $\Sigma$ is inductively defined by:
  
  \centerline{$E=\emptyset$, $E=\varepsilon$, $E=a$,}
  
  \centerline{$E=F+G$, $E=(F\cdot G)$, $E=(F^*)$,}
  
  \centerline{$E=\mathbb{F}_k(F)$}
  
   where $a$ is any symbol in $\Sigma$, $F$ and $G$ are any two AREs, $\mathbb{F}$ is any symbol-wise word comparison function and $k$ is any integer in $\mathbb{N}\cup\{\indef\}$.
 \end{definition}
 
 \begin{definition}
  The \emph{language denoted} by an ARE $E$ is the language $L(E)$ inductively defined by:
  
  \centerline{ $L(\emptyset)=\emptyset$, $L(\varepsilon)=\{\varepsilon\}$, $L(a)=\{a\}$,
  }
  
  \centerline{ $L(F+G)=L(F)\cup L(G)$, $L(F\cdot G)=L(F)\cdot L(G)$, $L(F^*)=L(F)^*$,}
  
  \centerline{
    $L(\mathbb{F}_k(F))=\mathbb{F}_k(L(F))$.
  }
  
   where $a$ is any symbol in $\Sigma$, $F$ and $G$ are any two AREs, $\mathbb{F}$ is any symbol-wise word comparison function and $k$ is any integer in $\mathbb{N}\cup\{\indef\}$.
 \end{definition}
 
  In order to prove that the language denoted by an ARE $E$ is regular, we will show how to compute a finite automaton recognizing $L(E)$.
  
\section{Hamming and Levenshtein Derivation Formulae}\label{se:hamLevDeriv}

In this section, we extend the derivation formulae to the family of approximate regular expressions where the word comparison functions are the usual Hamming and Levenshtein distances. Notice that the proofs are not given in this section, but will be stated in Section~\ref{se:lienEntreDeuxForm}, deduced from the proof of the general case provided in Section~\ref{se: deriv are}. 
  
  Let $a$ be a symbol in an alphabet $\Sigma$ and $L$ be a regular language over $\Sigma$. Let $k$ be an integer and $L'=\mathbb{L}_k(L)$. The quotient of $L'$ w.r.t. $a$ is by definition the set of words $w$ such that there exists a word $w'$ in $L'$ satisfying $\mathbb{L}(aw,w')\leq k$. Consequently, we distinguish the four following cases, according to the way $w'$ can be split:
  \begin{enumerate}
    \item $w'=aw''$ and $\mathbb{L}(a,a)+\mathbb{L}(w,w'')\leq k$: hence the word $w''$ is by definition in $a^{-1}(L)$ and $\mathbb{L}(w,w'')\leq k$. Consequently, $w\in \mathbb{L}_k(a^{-1}(L))$;
    \item $w'=bw''$ with $b\in\Sigma\setminus\{a\}$ and $\mathbb{L}(a,b)+\mathbb{L}(w,w'')\leq k$: hence the word $w''$ is by definition in $b^{-1}(L)$ and $\mathbb{L}(w,w'')\leq k-1$. Consequently, $w\in \mathbb{L}_{k-1}(b^{-1}(L))$;
    \item $\mathbb{L}(a,\varepsilon)+\mathbb{L}(w,w')\leq k$: hence the word $w'$ is by definition in $L$ and $\mathbb{L}(w,w')\leq k-1$. Consequently, $w\in \mathbb{L}_{k-1}(L)$;
    \item $w'=bw''$ with $b\in\Sigma$ and $\mathbb{L}(\varepsilon,b)+\mathbb{L}(aw,w'')\leq k$: hence the word $w''$ is by definition in $b^{-1}(L)$ and $\mathbb{L}(aw,w'')\leq k-1$. Consequently, $w\in a^{-1}(\mathbb{L}_{k-1}(b^{-1}(L)))$.
  \end{enumerate}
  
  Notice that for the Hamming distance, only the two first cases need to be considered since $\mathbb{H}(\alpha,\beta)=\indef$ whenever $\alpha=\varepsilon$ and $\beta\neq\varepsilon$ or $\alpha\neq \varepsilon$ and $\beta=\varepsilon$.
  
  As a consequence, the following lemma can be stated.  
  
  \begin{lemma}
    Let $L$ be a regular language over an alphabet $\Sigma$, $a$ be a symbol in $\Sigma$ and $k$ be an integer in $\mathbb{N}\cup\{\indef\}$. Then:
    
    \centerline{
      $a^{-1}(\mathbb{H}_k(L))= \mathbb{H}_k(a^{-1}(L)) \cup \bigcup_{b\in\Sigma\setminus\{a\}} \mathbb{H}_{k-1}(b^{-1}(L))$,
    }
    
    \centerline{
      $a^{-1}(\mathbb{L}_k(L))=
      \left(
          \begin{array}{l@{\ }l}
          & \mathbb{L}_k(a^{-1}(L)) \\
          \cup & \bigcup_{b\in\Sigma\setminus\{a\}} \mathbb{L}_{k-1}(b^{-1}(L))\\
          \cup & \mathbb{L}_{k-1}(L)\\
          \cup & a^{-1}(\bigcup_{b\in\Sigma}\mathbb{L}_{k-1}(b^{-1}(L))) \\
          \end{array}
          \right)$.
    }
  \end{lemma}
  
  In the remaining of this section, we consider restricted AREs that only use Hamming and Levenshtein distances.
  
  \begin{definition}
    Let $\Sigma$ be an alphabet.
    An \emph{Hamming-Levenshtein Approximate} \emph{Regular Expression} (HLARE)
    over $\Sigma$
    is an ARE  
    over $\Sigma$
    satisfying the following condition:
  
  \centerline{For any subexpression $G=\mathbb{F}_k(H)$, either $\mathbb{F}=\mathbb{H}$ or $\mathbb{F}=\mathbb{L}$.}
  \end{definition}

\subsection{Brzozowski Derivatives for an HLARE}

In this subsection, we extend the Brzozowski derivation to the HLAREs. From an HLARE $E$ and a word $w$, Brzozowski derivation allow us to syntactically compute an HLARE $D'_w(E)$, called the dissimilar derivative of $E$ w.r.t. $w$, denoting the language $w^{-1}(L(E))$.

\begin{definition}\label{def deriv diss hlare}
  Let $E$ be an HLARE over an alphabet $\Sigma$. Let $a$ and $b$ be two distinct symbols in $\Sigma$ and $w$ be a word in $\Sigma^*$.  The \emph{dissimilar derivative} of $E$ w.r.t. the symbol $a$ (resp. the word $w$) is the HLARE $D'_a(E)$ (resp. $D'_w(E)$) defined as follows:
  
  \centerline{
    $D'_a(\varepsilon)=D'_a(\emptyset)=D'_a(b)=\emptyset$,
    $D'_a(a)=\varepsilon$,
  }
  
  \centerline{
    $D'_a(E_1+E_2)=(D'_a(E_1)+D'_a(E_2))_{\sim_s}$,
    $D'_a(E_1^*)=D'_a(E_1)\cdot E_1^*$,
  }
  
  \centerline{
    $D'_a(E_1\cdot E_2)=
      \left\{
        \begin{array}{l@{\ }l}
          (D'_a(E_1)\cdot E_2 +D'_a(E_2))_{\sim_s} & \text{ if } \mathrm{r}(\varepsilon,E_1)=1,\\
          (D'_a(E_1)\cdot E_2)_{\sim_s} & \text{ if }\mathrm{r}(\varepsilon,E_1)=0,\\
        \end{array}
      \right.$
  }
  
  \centerline{
    $D'_a(\mathbb{H}_k(E_1))=( \mathbb{H}_k(D'_a(E_1)) + \sum_{b\in\Sigma\setminus\{a\}} \mathbb{H}_{k-1}(D'_b(E_1)))_{\sim_s}$,
  }
  
  \centerline{
    $D'_a(\mathbb{L}_k(E_1))=
      \left(
          \begin{array}{l@{\ }l}
          & \mathbb{L}_k(D'_a(E_1)) \\
          + & \sum_{b\in\Sigma\setminus\{a\}} \mathbb{L}_{k-1}(D'_b(E_1))\\
          + & \mathbb{L}_{k-1}(E_1)\\
          + & D'_a(\sum_{b\in\Sigma}\mathbb{L}_{k-1}(D'_b(E_1))) \\
          \end{array}
          \right)_{\sim_s}$,
  }
  
  \centerline{
    $D'_w(E)=
      \left\{
        \begin{array}{l@{\ }l}
          E & \text{ if } w=\varepsilon,\\          
          D'_{u}(D'_a(E)) & \text{ if }w=au\wedge a\in\Sigma\ \wedge\ u\in\Sigma^*,\\
        \end{array}
      \right.
    $
  }
  
  where $E_1$ and $E_2$ are any two HLARES and $k$ is any integer in $\mathbb{N}\cup\{\indef\}$.
\end{definition}

\begin{lemma}
  Let $E$ be an HLARE over an alphabet $\Sigma$. Let $w$ be a word in $\Sigma^*$. 
  Then:
  
  \centerline{
    $L(D'_w(E))=w^{-1}(L(E))$.
  }
\end{lemma}

Next lemma shows that the boolean $\mathrm{r}(\varepsilon,E)$ is syntactically computable for any HLARE $E$ using dissimilar derivatives.

\begin{lemma}
  Let $E=\mathbb{H}_k(E')$ and $F=\mathbb{L}_k(F')$ be two HLAREs over an alphabet $\Sigma$. Then the two following propositions are satisfied:
  \begin{itemize}
    \item $\varepsilon\in L(E) \Leftrightarrow \varepsilon \in L(E')$,
    \item $\varepsilon\in L(F) \Leftrightarrow \varepsilon \in L(F')\cup \bigcup_{a\in\Sigma} L(\mathbb{L}_{k-1}(D'_a(E')))$.
  \end{itemize}
\end{lemma}

Given an HLARE $E$, we denote by $\mathcal{D}_{HL}(E)$ the set $\{D'_w(E)\mid w\in\Sigma^*\}$ of the dissimilar derivatives of $E$.

\begin{lemma}
  The set $\mathcal{D}_{HL}(E)$ of dissimilar derivatives of an HLARE $E$ is finite.
\end{lemma}

From this finite set, one can compute a deterministic finite automaton that recognizes $L(E)$.

\begin{definition}
  Let $E$ be an HLARE over an alphabet $\Sigma$. The tuple $B'(E)=(\Sigma,Q,I,F,\delta)$ is defined by:
  \begin{itemize}
    \item $Q=\mathcal{D}_{HL}(E)$,
    \item $I=\{E\}$,
    \item $F=\{q\in Q\mid\mathrm{r}(\varepsilon,q)=1\}$,
    \item $\forall (q,a)\in Q\times \Sigma$, $\delta(q,a)=\{D'_a(q)\}$.
  \end{itemize}
\end{definition}

\begin{proposition}
  Let $E$ be an HLARE over an alphabet $\Sigma$. Then:
  
  \centerline{$B'(E)$ is a deterministic finite automaton that recognizes $L(E)$.}
\end{proposition}

  For any HLARE $E$, the automaton $B'(E)$ is called the \emph{dissimilar derivative finite automaton} of $E$.

Example~\ref{ex cons brzo} presents the computation of the dissimilar derivative automaton of an HLARE. Example~\ref{ex membership prob} illustrates the computation of the boolean $\mathrm{r}(w,E)$ for an HLARE $E$. Notice that in both of these examples, the following reductions are used:

\centerline{ $E+\emptyset=\emptyset+E=E$,}

\centerline{ $E\cdot\emptyset=\emptyset\cdot E=\emptyset$,}

\centerline{ $E\cdot\varepsilon=\varepsilon\cdot E=E$,}

\centerline{ $\mathbb{F}_\indef(E)=\emptyset$.}
  
  \begin{example}\label{ex cons brzo}
    Let $F=b^*(a+b)c^*$ and $E=\mathbb{H}_{1}(F)$ be an HLARE over $\Sigma=\{a,b,c\}$. The dissimilar derivatives of $E$ are the following expressions:

  \centerline{\begin{tabular}{l@{\ }||@{\ }l}
  \begin{tabular}{r@{\ }l@{\ }l}
    $D'_a(E)$ & $=\mathbb{H}_{0}(F)+\mathbb{H}_{1}(c^*)+\mathbb{H}_{0}(c^*)$ & $=E_1$\\
    $D'_b(E)$ & $=E+\mathbb{H}_{1}(c^*)+\mathbb{H}_{0}(c^*)$ & $=E_2$\\
    $D'_c(E)$ & $=\mathbb{H}_{0}(F)+\mathbb{H}_{0}(c^*)$ & $=E_3$\\
  \end{tabular}
  &
  \begin{tabular}{r@{\ }l@{\ }l}
    $D'_a(E_1)$ & $=\mathbb{H}_{0}(c^*)$ & $=E_4$\\
    $D'_b(E_1)$ & $=\mathbb{H}_{0}(F)+\mathbb{H}_{0}(c^*)$ & $=E_3$\\
    $D'_c(E_1)$ & $=\mathbb{H}_{1}(c^*)+\mathbb{H}_{0}(c^*)$ & $=E_5$\\
  \end{tabular}
  \\
  \begin{tabular}{r@{\ }l@{\ }l}
    $D'_a(E_2)$ & $=\mathbb{H}_{0}(F)+\mathbb{H}_{1}(c^*)+\mathbb{H}_{0}(c^*)$ & $=E_1$\\
    $D'_b(E_2)$ & $=E+\mathbb{H}_{1}(c^*)+\mathbb{H}_{0}(c^*)$ & $=E_2$\\
    $D'_c(E_2)$ & $=\mathbb{H}_{0}(F)+\mathbb{H}_{0}(c^*)+\mathbb{H}_{1}(c^*)$ & $=E_1$\\
  \end{tabular}
  &
  \begin{tabular}{r@{\ }l@{\ }l}
    $D'_a(E_3)$ & $=\mathbb{H}_{0}(c^*)$ & $=E_4$\\
    $D'_b(E_3)$ & $=\mathbb{H}_{0}(F)+\mathbb{H}_{0}(c^*)$ & $=E_3$\\
    $D'_c(E_3)$ & $=\mathbb{H}_{0}(c^*)$ & $=E_4$\\
  \end{tabular}
  \\
  \begin{tabular}{r@{\ }l@{\ }l}
    $D'_a(E_4)$ & $=\emptyset$\\
    $D'_b(E_4)$ & $=\emptyset$\\
    $D'_c(E_4)$ & $=\mathbb{H}_{0}(c^*)$ & $=E_4$\\
  \end{tabular}
  &
  \begin{tabular}{r@{\ }l@{\ }l}
    $D'_a(E_5)$ & $=\mathbb{H}_{0}(c^*)$ & $=E_4$\\
    $D'_b(E_5)$ & $=\mathbb{H}_{0}(c^*)$ & $=E_4$\\
    $D'_c(E_5)$ & $=\mathbb{H}_{1}(c^*)+\mathbb{H}_{0}(c^*)$ & $=E_5$\\
  \end{tabular}
  \\
  \end{tabular}}
  
  The dissimilar derivative automaton of $E$ is given Figure~\ref{fig aut brzo ham}.
  
  \end{example}
  
  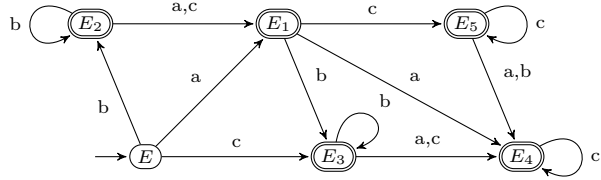
\begin{figure}[H]
    \centerline{ 
      \begin{tikzpicture}[node distance=2.5cm, bend angle=25]  
	    \node[state,initial] (q) [rounded rectangle] {$E$};
	    \node[state,accepting] (q1) [rounded rectangle, above right of=q] {$E_1$};   
	    \node[state,accepting] (q2) [rounded rectangle,left of=q1] {$E_2$};
	    \node[state,accepting] (q3) [rounded rectangle, right of=q] {$E_3$};
	    \node[state,accepting] (q4) [rounded rectangle, right of=q3] {$E_4$};	
	    \node[state,accepting] (q5) [rounded rectangle, right of=q1] {$E_5$};	    	  
	    \path[->]
	      (q)  edge node {a} (q1)
	           edge node {b} (q2)	
	           edge node {c} (q3) 
	      (q1) edge node {a} (q4)
	           edge node {b} (q3)	
	           edge node {c} (q5) 
	      (q2) edge node {a,c} (q1)
	           edge [in=-150,out=150,loop,swap] node {b} ()	
	      (q3) edge node {a,c} (q4)
	           edge [in=20,out=80,loop] node {b} ()	
	      (q4) edge [in=-30,out=30,loop] node {c} ()	
	      (q5) edge node {a,b} (q4)
	           edge [in=-30,out=30,loop] node {c} ();	            
      \end{tikzpicture}
    }   
    \caption{The dissimilar derivative automaton of $E=\mathbb{H}_{1}(b^*(a+b)c^*)$}
    \label{fig aut brzo ham}
  \end{figure}
  
  \begin{example}\label{ex membership prob}
    Let $G=\mathbb{L}_1((aba+abb)a(a)^*)$ be an HLARE over the alphabet $\Sigma=\{a,b\}$ and $w=aba$ be a word in $\Sigma^*$. Then: $\mathrm{r}(aba,G)=\mathrm{r}(\varepsilon,D'_{aba}(G))$. Let us first compute the HLARE $D'_{aba}(G)$:
    
    \begin{tabular}{l@{\ }l}    
      $D'_a(G)=$ & $\mathbb{L}_1((ba+bb)a(a)^*)+\mathbb{L}_0((aba+abb)a(a)^*)=G_1$\\
      $D'_b(G_1)=$ & $\mathbb{L}_1((a+b)a(a)^*)+\mathbb{L}_0((ba+bb)a(a)^*)+\mathbb{L}_0(a(a)^*)=G_2$\\
      $D'_a(G_2)=$ & $\mathbb{L}_1(a(a)^*)+\mathbb{L}_0(a(a)^*)+\mathbb{L}_0((a+b)a(a)^*)+\mathbb{L}_0((a)^*)=G_3$\\
    \end{tabular}
    
    \noindent Hence $\mathrm{r}(aba,G)=\mathrm{r}(\varepsilon,G_3)$.
    Furthermore, since $\varepsilon\in L( \mathbb{L}_0(D'_a(a(a)^*))) $, it holds that $\varepsilon\in L(\mathbb{L}_1(a(a)^*))$. Consequently, $\mathrm{r}(\varepsilon,G_3)=1$ and $aba$ belongs to $L(G)$.
    
     Notice that in this case:
     \begin{enumerate}
       \item The word $w$ is split up into $s_w=(a,b,a,\varepsilon)$;
       \item The word $w'=abaa$ in $L((aba+abb)a(a)^*)$ can be split up into $s_{w'}=(a,b,a,a)$; 
       \item It holds $\mathcal{L}(s_w,s_{w'})=1$.
     \end{enumerate}
     
     Another split-up is presented in Example~\ref{ex membership prob anti}.
  \end{example}

\subsection{Antimirov Partial  Derivatives of an HLARE}

In this subsection, we extend the Antimirov derivation to the HLAREs. From an HLARE $E$ and a word $w$, Antimirov derivation allows us to compute a set $\Delta_w(E)$ of HLAREs, called the partial derivative of $E$ w.r.t. $w$. Any HLARE in $\Delta_w(E)$ is called a derivated term of $E$ w.r.t. $w$. Finally, we state that the union of the languages denoted by the derivated terms in $\Delta_w(E)$ is equal to $w^{-1}(L(E))$.

\begin{definition}\label{def deriv part hlare}
  Let $E$ be an HLARE over an alphabet $\Sigma$. Let $a$ and $b$ be two distinct symbols in $\Sigma$ and $w$ be a word in $\Sigma^*$.  The \emph{partial derivative} of $E$ w.r.t. the symbol $a$ (resp. to the word $w$) is the set $\Delta_a(E)$ (resp. $\Delta_w(E)$) of HLAREs defined as follows:
  
  \centerline{
    $\Delta_a(\varepsilon)=\Delta_a(\emptyset)=\Delta_a(b)=\emptyset$,
    $\Delta_a(a)=\{\varepsilon\}$,
  }
  
  \centerline{
    $\Delta_a(E_1+E_2)=\Delta_a(E_1)\cup \Delta_a(E_2)$,
    $\Delta_a(E_1^*)=\Delta_a(E_1)\cdot E_1^*$,
  }
  
  \centerline{
    $\Delta_a(E_1\cdot E_2)=
      \left\{
        \begin{array}{l@{\ }l}
          \Delta_a(E_1)\cdot E_2 \cup \Delta_a(E_2)  & \text{ if } \mathrm{r}(\varepsilon,E_1)=1,\\
          \Delta_a(E_1)\cdot E_2  & \text{ if }\mathrm{r}(\varepsilon,E_1)=0,\\
        \end{array}
      \right.$
  }
  
  \centerline{
    $\Delta_a(\mathbb{H}_k(E_1))= \mathbb{H}_k(\Delta_a(E_1)) \cup \bigcup_{b\in\Sigma\setminus\{a\}} \mathbb{H}_{k-1}(\Delta_b(E_1)) $,
  }
  
  \centerline{
    $\Delta_a(\mathbb{L}_k(E_1))=
      \left(
          \begin{array}{l@{\ }l}
          & \mathbb{L}_k(\Delta_a(E_1)) \\
          \cup & \bigcup_{b\in\Sigma\setminus\{a\}} \mathbb{L}_{k-1}(\Delta_b(E_1))\\
          \cup & \{\mathbb{L}_{k-1}(E_1)\}\\
          \cup & \Delta_a(\bigcup_{b\in\Sigma}\mathbb{L}_{k-1}(\Delta_b(E_1)) \\
          \end{array}
          \right)$,
  }
  
  \centerline{
    $\Delta_w(E)=
      \left\{
        \begin{array}{l@{\ }l}
          \{E\} & \text{ if } w=\varepsilon,\\          
          \Delta_{w'}(\Delta_a(E)) & \text{ if }w=aw'\wedge a\in\Sigma\ \wedge\ w'\in\Sigma^*,\\
        \end{array}
      \right.
    $
  }
  
  where $E_1$ and $E_2$ are any two HLARES and $k$ an integer in $\mathbb{N}\cup\{\indef\}$ and where for any set $\mathcal{E}$ of HLAREs, for any HLARE $F$, for any symbol $a$ in $\Sigma$,
  
  \centerline{$\mathcal{E}\cdot F=\bigcup_{E\in\mathcal{E}}\{E\cdot F\}$,}
  
  \centerline{$\Delta_a(\mathcal{E})=\bigcup_{E\in\mathcal{E}}\Delta_a(E)$,}
  
  \centerline{$\mathbb{H}_k(\mathcal{E})=\bigcup_{E\in\mathcal{E}}\{\mathbb{H}_k(E)\}$,}
  
  \centerline{$\mathbb{L}_k(\mathcal{E})=\bigcup_{E\in\mathcal{E}}\{\mathbb{L}_k(E)\}$.}
\end{definition}

\begin{lemma}
  Let $E$ be an HLARE over an alphabet $\Sigma$. Let $w$ be a word in $\Sigma^*$. 
  Then:
  
  \centerline{
    $\bigcup_{G\in \Delta_w(E)}L(G)=w^{-1}(L(E))$.}
\end{lemma}

Next lemma shows that the boolean $\mathrm{r}(\varepsilon,E)$ is syntactically computable for any HLARE $E$ using partial derivation.

\begin{lemma}
  Let $E=\mathbb{H}_k(E')$ and $F=\mathbb{L}_k(F')$ be two HLAREs over an alphabet $\Sigma$. Then the two following conditions are satisfied:
  \begin{itemize}
    \item $\varepsilon\in L(E) \Leftrightarrow \varepsilon \in L(E')$,
    \item $\varepsilon\in L(F) \Leftrightarrow \varepsilon \in L(F')\cup \bigcup_{a\in\Sigma,G\in\Delta_a(F')} L(\mathbb{L}_{k-1}(G))$.
  \end{itemize}
\end{lemma}

Given an HLARE $E$, we denote by $\mathcal{DT}_{HL}(E)$ the set $\bigcup_{w\in\Sigma^*} \Delta_w(E)$ of the derivated terms of $E$.

\begin{lemma}
  The set $\mathcal{DT}_{HL}(E)$ of the derivated terms of an HLARE $E$ is finite.
\end{lemma}

From this finite set, one can compute a finite automaton that recognizes $L(E)$.

\begin{definition}
  Let $E$ be an HLARE over an alphabet $\Sigma$. The tuple $A(E)=(\Sigma,Q,I,F,\delta)$ is defined by:
  \begin{itemize}
    \item $Q=\mathcal{DT}_{HL}(E)$,
    \item $I=\{E\}$,
    \item $F=\{q\in Q\mid\mathrm{r}(\varepsilon,q)=1\}$,
    \item $\forall (q,a)\in Q\times \Sigma$, $\delta(q,a)=\Delta_a(q)$.
  \end{itemize}
\end{definition}

\begin{proposition}
  Let $E$ be an HLARE over an alphabet $\Sigma$. Then:
  
  \centerline{$A(E)$ is a finite automaton that recognizes $L(E)$.}
\end{proposition}

  For any HLARE $E$, the automaton $A(E)$ is the \emph{derivated term finite automaton} of $E$.

Example~\ref{ex cons anti} presents the computation of the derivated term automaton of an HLARE. Example~\ref{ex membership prob anti} illustrates the computation of the boolean $\mathrm{r}(w,E)$ for an HLARE $E$. Notice that in both of these examples, the five following reductions are used:

\centerline{ $E+\emptyset=\emptyset+E=E$,}

\centerline{ $E\cdot\emptyset=\emptyset\cdot E=\emptyset$,}

\centerline{ $E\cdot\varepsilon=\varepsilon\cdot E=E$,}

\centerline{ $\mathbb{F}_\indef(E)=\emptyset$,}

\centerline{ $\{\emptyset\}=\emptyset$.\footnote{The four first equalities are HLAREs reductions whereas the last one is a HLARE set reduction.}}

  \begin{example}\label{ex cons anti}
    Let $E$ be the HLARE defined in Example~\ref{ex cons brzo}. The partial derivatives of $E$ are the following sets of expressions:
  
  \centerline{\begin{tabular}{r@{\ }l||r@{\ }l}
    $\Delta_a(E)$ & $=\{\mathbb{H}_{0}(F),\mathbb{H}_{1}(c^*),\mathbb{H}_{0}(c^*)\}$ & $\Delta_a(\mathbb{H}_1(c^*))$ & $=\{\mathbb{H}_{0}(c^*)\}$\\
    $\Delta_b(E)$ & $=\{E,\mathbb{H}_{1}(c^*),\mathbb{H}_{0}(c^*)\}$ & $\Delta_b(\mathbb{H}_1(c^*))$ & $=\{\mathbb{H}_{0}(c^*)\}$\\
    $\Delta_c(E)$ & $=\{\mathbb{H}_{0}(F),\mathbb{H}_{0}(c^*)\}$ & $\Delta_c(\mathbb{H}_1(c^*))$ & $=\{\mathbb{H}_{1}(c^*)\}$\\
    $\Delta_a(\mathbb{H}_0(F))$ & $=\{\mathbb{H}_{0}(c^*)\}$ & $\Delta_a(\mathbb{H}_0(c^*))$ & $=\emptyset$\\
    $\Delta_b(\mathbb{H}_0(F))$ & $=\{\mathbb{H}_{0}(F),\mathbb{H}_{0}(c^*)\}$ & $\Delta_b(\mathbb{H}_0(c^*))$ & $=\emptyset$\\
    $\Delta_c(\mathbb{H}_0(F))$ & $=\emptyset$ & $\Delta_c(\mathbb{H}_0(c^*))$ & $=\{\mathbb{H}_{0}(c^*)\}$\\
  \end{tabular}}
  
  The derivated term automaton of $E$ is given Figure~\ref{fig aut ant ham}.
  
  \end{example}
  
  \begin{figure}[H]
    \centerline{ 
      \begin{tikzpicture}[node distance=2.5cm]  
	    \node[state,initial] (q1) [rounded rectangle] {$\mathbb{H}_{1}(F)$};
	    \node[state] (q2) [rounded rectangle, above  of=q1,node distance=1.5cm] {$\mathbb{H}_{0}(F)$};   
	    \node[state,accepting] (q3) [rounded rectangle,  right of=q1,node distance=4cm] {$\mathbb{H}_{1}(c^*)$};
	    \node[state,accepting] (q4) [rounded rectangle, right  of=q2,node distance=4cm] {$\mathbb{H}_{0}(c^*)$};	    	  
	    \path[->]
	      (q1)  edge [swap]node {a,c} (q2)
	            edge [swap] node {a,b} (q3)	
	            edge node {a,b,c} (q4)	      	 
	            edge [in=135,out=195,loop] node {b} ()	 
	      (q2)  edge node {a,b} (q4)
	            edge [in=170,out=230,loop] node {b} ()	 
	      (q3)  edge [swap] node {a,b} (q4)
	            edge [in=30,out=-30,loop,swap] node {c} ()	 
	      (q4)  edge [in=30,out=-30,loop,swap] node {c} ();	            
      \end{tikzpicture}
    }   
    \caption{The derivated term automaton of $E=\mathbb{H}_{1}(b^*(a+b)c^*)$}
    \label{fig aut ant ham}
  \end{figure}
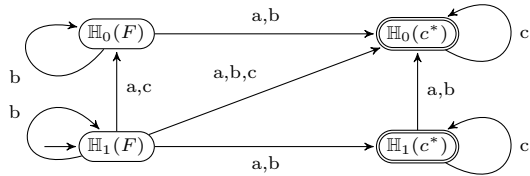
  
  \begin{example}\label{ex membership prob anti}
    Let $G=\mathbb{L}_1((aba+abb)a(a)^*)$ be the HLARE defined in Example~\ref{ex membership prob} and $w=aba$ be a word in $\Sigma^*$. Then: $\mathrm{r}(aba,G)=\bigvee_{H\in\Delta_{aba}(G)}\mathrm{r}(\varepsilon,H)$. Let us first compute the HLARE set $\Delta_{aba}(G)$:
    
    \begin{tabular}{l@{\ }l}    
      $\Delta_a(G)$ & $=\{\mathbb{L}_1(baa(a)^*),\mathbb{L}_1(bba(a)^*),\mathbb{L}_0((aba+abb)a(a)^*)\}$\\
      & $=\mathcal{G}_1$\\  
      $\Delta_b(\mathcal{G}_1)$ & $=\{\mathbb{L}_1(aa(a)^*),\mathbb{L}_0(baa(a)^*),
      \mathbb{L}_1(ba(a)^*),\mathbb{L}_0(bba(a)^*),\mathbb{L}_0(a(a)^*)\}$\\
      & $=\mathcal{G}_2$\\ 
      $\Delta_a(\mathcal{G}_2)$ & $=\{\mathbb{L}_1(a(a)^*),\mathbb{L}_0(aa(a)^*),\mathbb{L}_0((a)^*),\mathbb{L}_0(ba(a)^*),\mathbb{L}_0(a(a)^*)\}$\\
      & $=\mathcal{G}_3$\\ 
    \end{tabular}
    
    \noindent Hence $\mathrm{r}(aba,G)=\bigvee_{H\in\mathcal{G}_3}\mathrm{r}(\varepsilon,H)$.
    Furthermore, since $\varepsilon\in L( \mathbb{L}_0((a)^*)) $, it holds that $\mathrm{r}(\varepsilon,G_3)=1$. Finally, $aba$ belongs to $L(G)$.    
    
     Notice that in this case:
     \begin{enumerate}
       \item The word $w$ is split up into $s_w=(a,b,\varepsilon,a)$;
       \item The word $w'=abaa$ in $L((aba+abb)a(a)^*)$ can be split up into $s_{w'}=(a,b,a,a)$; 
       \item It holds $\mathcal{L}(s_w,s_{w'})=1$.
     \end{enumerate}
     
     Another split-up is presented in Example~\ref{ex membership prob}.  
  \end{example}

\section{Word Comparison Functions, Quotients and Derivatives}\label{se: deriv are}

  In this section, we address the general case. We present two constructions of an automaton from an ARE using Brzozowski's derivatives and Antimirov's ones, respectively leading to a deterministic automaton and a non-deterministic one.
  We first show how to compute the quotient of a given language $\mathbb{F}_k(L)$ w.r.t. a symbol $a$, where $\mathbb{F}$ is a given word comparison function, $k$ is an integer and $L$ is a regular language.

  \subsection{Quotient of a Language}

Let $\mathbb{F}$ be a word comparison function associated with a symbol-wise sequence comparison function $\mathcal{F}$
defined over an alphabet $\Sigma$.
Let $k$ be a positive integer, $a$ be a symbol in $\Sigma$, $u=aw$ be a word of $\Sigma^+$, and $L'$ be a regular language of $\Sigma^*$.
According to Definition~\ref{def fkl}, the word $u$ is in $L=\mathbb{F}_k(L')$
if and only if there exists a word $v\in L'$ such that $\mathbb{F}(u, v)\leq k$.
According to Definition~\ref{def split up comp func},
this is equivalent to the existence of an alignment\footnote{An alignment between two words $u$ and $v$ is a pair $(s,s')$ of sequences of same size such that $s\in\mathrm{Split}(u)$ and $s'\in\mathrm{Split}(v)$.}
$(u',v') \in \mathrm{Split}_{n}(u)\times \mathrm{Split}_{n}(v)$,
where $n$ is a positive integer,
between $u$ and $v$,
the cost $\mathcal{F}(u',v')$ of which is not greater than $k$.
Let $u'=(u'_1,\ldots,v'_n)$ and $v'=(v'_1,\ldots,v'_n)$.
\textbf{(a)} If $n=1$, $\mathbb{F}(u,v)=\mathcal{F}((a),(v'_1))$ and since $u=aw$, $a\in L \Leftrightarrow w\in \mathbb{F}_{k-\mathcal{F}((a),(v'_1))}{v'_1}^{-1}(L')$. 
\textbf{(b)} Otherwise, let us set $u''=(u'_2,\ldots,u'_n)$ and $v''=(v'_2,\ldots,v'_n)$. Moreover, let us set $t=u$ if $u'_1=\varepsilon$ and $t=u'_2\cdots u'_n$ otherwise; let us similarly set $z=v$ if $v'_1=\varepsilon$ and $z=v'_2\cdots v'_n$ otherwise. Obviously, the word $z$ belongs to ${v'_1}^{-1}(L')$. Since $\mathbb{F}$ is a symbol-wise word comparison function,
there exists an alignment $(u',v')$ between $u$ and $v$ satisfying $\mathcal{F}(u',v')\leq k$
if and only if
there exists an alignment $(u'',v'')$ between $t$
 and $z$
satisfying $\mathcal{F}(u'',v'')\leq k-\mathcal{F}((u'_1),(v'_1))$.
According to Definition~\ref{def split up comp func},
this is equivalent to the existence of a word $z\in {v'_1}^{-1}(L')$
such that $\mathbb{F}(t,z)\leq k-\mathcal{F}((u'_1),(v'_1))$.
According to Definition~\ref{def fkl}, it is equivalent to say that the word $t$ is in 
$\mathbb{F}_{k-\mathcal{F}((u'_1),(v'_1))}({v'_1}^{-1}(L'))$.
Depending on the value of $(u'_1,v'_1)$ we can distinguish the following cases:\\
{\bf Case 1} $(u'_1,v'_1)=(a,b)$, with $b \in \Sigma$:
$u=aw\in L \Leftrightarrow w \in \mathbb{F}_{k-\mathcal{F}(a,b)}(b^{-1}L')$,\\  
{\bf Case 2} $(u'_1,v'_1)=(a,\varepsilon)$ with $a\in\Sigma$:
$u=aw\in L \Leftrightarrow w \in \mathbb{F}_{k-\mathcal{F}(a,\varepsilon)}(L')$,\\  
{\bf Case 3} $(u'_1,v'_1)=(\varepsilon,b)$, with $b \in \Sigma$:
$u=aw\in L \Leftrightarrow w \in a^{-1}(\mathbb{F}_{k-\mathcal{F}(\varepsilon,b)}(b^{-1}L'))$.
Since $w \in a^{-1}\mathbb{F}_k(L')$ $\Leftrightarrow $ $aw \in \mathbb{F}_k(L')$,
the three previous cases provide a recursive expression of the quotient of 
the language $\mathbb{F}_k(L')$ w.r.t. a symbol $a \in \Sigma$.
Unfortunately, its computation may imply a recursive loop, due to Case 3, when 
$\mathcal{F}((\varepsilon),(b))=0$.
It is possible to get rid of this loop by precomputing the set of all the quotients of $L'$ w.r.t. words $w$ such that $\mathbb{F}(\varepsilon,w)=0$.
In this purpose, let us set $\mathcal{W}_\mathcal{F}=(\bigcup_{b\in\Sigma,\mathcal{F}((\varepsilon),(b))=0} \{b\})^*$ 
and $X(L')=\{L'\}\cup\bigcup_{w\in\mathcal{W}_\mathcal{F}}\{w^{-1}(L')\}$. Let us notice that if $L'$ is a regular language, the set of its residuals is finite; as a consequence, so is $X(L')$.

  \begin{lemma}\label{lem quot app lang}
    Let $L=\mathbb{F}_{k}(L')$ be a language over an alphabet $\Sigma$ where $L'$ is a regular language, $\mathbb{F}$ is a symbol-wise word comparison function associated with a sequence comparison function $\mathcal{F}$ and $a$ be a symbol in $\Sigma$. The \emph{quotient} of $L$ w.r.t. $a$ is the language $a^{-1}(L)$ computed as follows:
    
    \centerline{$a^{-1}(L)=
      \left\{
        \begin{array}{l@{\ }l}
          & \bigcup_{L''\in X(L'),b\in\Sigma}( \mathbb{F}_{k-\mathcal{F}((a),(b))}(b^{-1}(L'')))
           \cup \    \bigcup_{L''\in X(L')} \mathbb{F}_{k-\mathcal{F}((a),(\varepsilon))}(L'')\\
          \cup & a^{-1}(\bigcup_{L''\in X(L'),b \in\Sigma,\mathcal{F}((\varepsilon),(b))\neq 0}( \mathbb{F}_{k-\mathcal{F}((\varepsilon),(b))}(b^{-1}(L''))))\\
        \end{array}
      \right.
    $}
  
  \centerline{
    where $X(L')=\{L'\}\cup\bigcup_{w\in\mathcal{W}_\mathcal{F}}w^{-1}(L')$ with $\mathcal{W}_\mathcal{F}=(\bigcup_{b\in\Sigma,\mathcal{F}((\varepsilon),(b))=0} \{b\})^*$.
  }
  \end{lemma}
  \begin{proof}
    For any symbol $\alpha,\beta$ in $\Sigma\cup\{\varepsilon\}$, let us set $k_{\alpha,\beta}=k-\mathcal{F}((\alpha),(\beta))$.
    
    $u\in a^{-1}(L)$ $\Leftrightarrow$ $au\in L$ $\Leftrightarrow$ $\exists w\in L', \mathbb{F}(au,w)\in\Intervalle{0}{k}$
    
    \noindent $\Leftrightarrow$ 
      $\left\{
        \begin{array}{l@{\ }l}
          & \exists b\in\Sigma, \exists w_1bw_2\in L', \mathbb{F}(\varepsilon,w_1)=0\wedge \mathbb{F}(u,w_2)\leq k_{a,b}\\
          \vee & \exists w_1w_2\in L', \mathbb{F}(\varepsilon,w_1)=0\wedge \mathbb{F}(u,w_2)\leq k_{a,\varepsilon}\\
          \vee & \exists b\in\Sigma, \exists w_1bw_2\in L', \mathbb{F}(\varepsilon,w_1)=0\wedge \mathcal{F}((\varepsilon),(b))\neq 0 \wedge \mathbb{F}(au,w_2)\leq k_{\varepsilon,b}\\
        \end{array}
      \right.$
      
      \vspace{\baselineskip}
    
    \noindent $\Leftrightarrow$ 
      $\left\{
        \begin{array}{l@{\ }l}
          & \exists b\in\Sigma, \exists w_1\in\mathcal{W}_\mathcal{F}, \exists w_2\in (w_1b)^{-1}(L'), \mathbb{F}(u,w_2)\leq k_{a,b}\\
          \vee & \exists w_1\in\mathcal{W}_\mathcal{F}, \exists w_2\in (w_1)^{-1}(L'), \mathbb{F}(u,w_2)\leq k_{a,\varepsilon}\\
          \vee & \exists b\in\Sigma, \exists w_1\in\mathcal{W}_\mathcal{F} , \exists w_2\in (w_1b)^{-1}L',\\
          & \mathcal{F}((\varepsilon),(b))\neq 0 \wedge \mathbb{F}(au,w_2)\leq k_{\varepsilon,b}\\
        \end{array}
      \right.$
      
      \vspace{\baselineskip}
    
    \noindent $\Leftrightarrow$ 
      $\left\{
        \begin{array}{l@{\ }l}
          & \exists b\in\Sigma, \exists w_2\in b^{-1}(\bigcup_{L''\in X(L')}L''), \mathbb{F}(u,w_2)\leq k_{a,b}\\
          \vee & \exists w_2\in \bigcup_{L''\in X(L')}L'', \mathbb{F}(u,w_2)\leq k_{a,\varepsilon}\\
          \vee & \exists b\in\Sigma, \exists w_2\in b^{-1}(\bigcup_{L''\in X(L')}L''), \mathcal{F}((\varepsilon),(b))\neq 0 \wedge \mathbb{F}(au,w_2)\leq k_{\varepsilon,b}\\
        \end{array}
      \right.$
      
      \vspace{\baselineskip}
    
    \noindent $\Leftrightarrow$ 
      $\left\{
        \begin{array}{l@{\ }l}
          & \exists b\in\Sigma, u\in \mathbb{F}_{k_{a,b}} \bigcup_{L''\in X(L')}b^{-1}(L'')\\
          \vee & u \in \bigcup_{L''\in X(L')}\mathbb{F}_{k_{a,\varepsilon}}(L'')\\
          \vee & \exists b\in\Sigma, au \in \mathbb{F}_{k_{\varepsilon,b}}(\bigcup_{L''\in X(L')}b^{-1}(L''))\\
        \end{array}
      \right.$
      
      \vspace{\baselineskip}
    
    \noindent $\Leftrightarrow$ 
      $\left\{
        \begin{array}{l@{\ }l}
          & u\in \bigcup_{L''\in X(L'),b\in\Sigma} \mathbb{F}_{k-\mathcal{F}((a),(b))} b^{-1}(L'')\\
          \vee & u \in \bigcup_{L''\in X(L')} \mathbb{F}_{k-\mathcal{F}((a),(\varepsilon))}(L'')\\
          \vee & u \in a^{-1}(\bigcup_{L''\in X(L'),b \in\Sigma,\mathcal{F}((\varepsilon),(b))\neq 0} \mathbb{F}_{k-\mathcal{F}((\varepsilon),(b))}b^{-1}(L''))\\
        \end{array}
      \right.$   
    %\cqfd
  \end{proof}
  
  \subsection{Brzozowski Derivatives for an ARE}
  
  An extension of Brzozowski derivatives can be directly deduced from the computation of the quotient presented in Lemma~\ref{lem quot app lang}.
  
  \begin{definition}\label{def diss deriv are}
    Let $E=\mathbb{F}_{k}(E')$ be an ARE over an alphabet $\Sigma$ where $\mathbb{F}$ is associated with $\mathcal{F}$ and $a$ be a symbol in $\Sigma$. The \emph{dissimilar derivative} of $E$ w.r.t. $a$ is the expression $\frac{d'}{d'_a}(E)$ defined by:
    
    \centerline{$\frac{d'}{d'_a}(E)=
      \left(
        \begin{array}{l@{\ }l}
          & \sum_{F\in X(E'),b\in\Sigma}( \mathbb{F}_{k-\mathcal{F}((a),(b))}(\frac{d'}{d'_b}(F)))\\
          + &   \sum_{F\in X(E')} \mathbb{F}_{k-\mathcal{F}((a),(\varepsilon))}(F)\\
          + & \frac{d'}{d'_a}(\sum_{F\in X(E'),b \in\Sigma,\mathcal{F}((\varepsilon),(b))\neq 0}( \mathbb{F}_{k-\mathcal{F}((\varepsilon),(b))}(\frac{d'}{d'_b}(F))))\\
        \end{array}
      \right)_{\sim_s}
    $,}
  
  \centerline{
    where $X(E')=\{E'\}\cup\bigcup_{w\in\mathcal{W}_\mathcal{F}}\frac{d'}{d'_w}(E')$ with $\mathcal{W}_\mathcal{F}=(\bigcup_{b\in\Sigma,\mathcal{F}((\varepsilon),(b))=0} \{b\})^*$.
  }
  \end{definition}

  Let us show that the set of dissimilar derivatives of any HLARE $E$ is 
  finite
   (Lemma~\ref{lem deriv ham fini}), that 
   the dissimilar derivative of $E$ w.r.t. a word $w$ denotes the quotient of $L(E)$ w.r.t. $w$ (Lemma~\ref{lem form deriv ham ok}) and  how to determine whether the empty word belongs to the language denoted by $E$ (Lemma~\ref{lem eps are}).
  
  \begin{lemma}\label{lem deriv ham fini}
    Let $E=\mathbb{F}_{k}(E')$ be an ARE over an alphabet $\Sigma$ and $\mathcal{D}_E$ be the set of dissimilar derivatives of $E$. Then $\mathcal{D}_E$ is a finite set of AREs. Moreover, its computation halts.
  \end{lemma}
  \begin{proof}
    Consider that $\mathbb{F}$ is associated with $\mathcal{F}$. Let us show by induction over the structure of $E'$ and by recurrence over $k$ that $\mathcal{D}_E$ is a finite set of AREs.
  
  By induction, the set $\mathcal{D}_{E'}$ is a finite set of AREs. Consequently, since $X(E')$ is a subset of $\mathcal{D}_{E'}$, \textbf{(Fact~1)} $X(E')$ is a finite set of derivatives of $E'$.
  
  In order to show that $\mathcal{D}_{E}$ is a finite set, let us show that any derivative $G$ of $E$ satisfies the property $\mathrm{P}(E',k)$: $G$ is a finite sum of expressions of type $\mathbb{F}_{k'}(G')$ with $k'\leq k$ and $G'$ a derivative of $E'$.
  
  According to \textbf{Fact~1}, any subexpression $\mathbb{F}_{k-\mathcal{F}((a),(\varepsilon))}(F)$ with $F\in X(E')$ satisfies $\mathrm{P}(E',k)$. Since $X(E')$ is a subset of $\mathcal{D}_{E'}$, $\frac{d'}{d'_b}(F)$ is a derivative of $E'$ for any $b$ in $\Sigma$. Consequently, the expression $\sum_{F\in X(E'),b\in\Sigma}( \mathbb{F}_{k-\mathcal{F}((a),(b))}(\frac{d'}{d'_b}(F)))$ also satisfies $\mathrm{P}(E',k)$.
  Finally, by recurrence hypothesis, for $k'< k$, any derivative of an expression $\mathbb{F}_{k'}(G')$ satisfies $\mathrm{P}(G',k')$. Consequently, any derivative of $\mathbb{F}_{k-\mathcal{F}((\varepsilon),(b))}(\frac{d'}{d'_b}(F))$ satisfies $\mathrm{P}(\frac{d'}{d'_b}(F),k-\mathcal{F}((\varepsilon),(b)))$ if $\mathcal{F}((\varepsilon),(b))\neq 0$. Since $F$ is a derivative of $E'$, so is $\frac{d'}{d'_b}(F)$, and since $k-\mathcal{F}((\varepsilon),(b))< k$, any derivative of $\mathbb{F}_{k-\mathcal{F}((\varepsilon),(b))}(\frac{d'}{d'_b}(F))$ satisfies $\mathrm{P}(E',k)$. As a consequence, \textbf{(Fact~2)} any derivative of $E$ w.r.t. a symbol $a$ satisfies $\mathrm{P}(E',k)$.
  
  Let us show now that if an expression $H$ satisfies $\mathrm{P}(E',k)$, then any symbol derivative of $H$ also satisfies $\mathrm{P}(E',k)$. Since $H$ is a sum of expressions of type $\mathbb{F}_{k'}(G')$ where $k'\leq k$ and $G'$ is a derivative of $E'$, any symbol derivative $H'$ of $H$ is the sum of the derivatives of the expressions $H$ is the sum of.
  According to \textbf{Fact~2}, any symbol derivative of an expression $\mathbb{F}_{k'}(G')$ satisfies $\mathrm{P}(G',k')$. Since $G'$ is a derivative of $E'$ and $k'\leq k$, any expression satisfying $\mathrm{P}(G',k')$ also satisfies $\mathrm{P}(E',k)$. As a consequence, any derivative of $E$ w.r.t. a word $w$ in $\Sigma^*$ satisfies $\mathrm{P}(E',k)$.
  
  As a conclusion, since any derivative of $E$ is a sum of expressions all belonging to the finite set $\{\mathbb{F}_{k'}(G)\mid k'\leq k \wedge G\in\mathcal{D}_{E'}\}$, using the ACI-equivalence, $\mathcal{D}_E$ is a finite set of AREs. Moreover, by induction over $E'$ and by recurrence over $k$, since any derivative of an expression $F$ in $X(E')$ belongs to the finite set of derivatives of $E'$ the computation of which halts, and since $\mathcal{F}((\varepsilon),(b))\neq 0$ implies that $k-\mathcal{F}((\varepsilon),(b))<k$, the computation of $\mathcal{D}_{E}$ halts.  
    %\cqfd
  \end{proof}
  
  \begin{lemma}\label{lem form deriv ham ok}
    Let $E=\mathbb{F}_{k}(E')$ be an ARE over an alphabet $\Sigma$ and $a$ be a symbol in $\Sigma$. Then $L(\frac{d'}{d'_a}(E))=a^{-1}(L(E))$.
  \end{lemma}
  \begin{proof} 
    By induction over the structure of $E$. According to Lemma~\ref{lem quot app lang}:
    
    \centerline{$a^{-1}(L(E))=
      \left\{
        \begin{array}{l@{\ }l}
          & \bigcup_{L''\in X(L(E')),b\in\Sigma}( \mathbb{F}_{k-\mathcal{F}((a),(b))}(b^{-1}(L'')))\\
          \cup &   \bigcup_{L''\in X(L(E'))} \mathbb{F}_{k-\mathcal{F}((a),(\varepsilon))}(L'')\\
          \cup & a^{-1}(\bigcup_{L''\in X(L(E')),b \in\Sigma,\mathcal{F}((\varepsilon),(b))\neq 0}( \mathbb{F}_{k-\mathcal{F}((\varepsilon),(b))}(b^{-1}(L''))))\\
        \end{array}
      \right.
    $,}
  
  \centerline{
    where $X(L(E'))=\{L(E')\}\cup\bigcup_{w\in\mathcal{W}_\mathcal{F}}w^{-1}(L(E'))$ with $\mathcal{W}_\mathcal{F}=(\bigcup_{b\in\Sigma,\mathcal{F}((\varepsilon),(b))=0} \{b\})^*$.
  }
  
  Let $X(E')=\{E'\}\cup\bigcup_{w\in\mathcal{W}_\mathcal{F}}\frac{d'}{d'_w}(E')$.  
  By induction over $E'$, for any word $w$ in $\Sigma^*$, $w^{-1}(L(E'))=L(\frac{d'}{d'_w}(E'))$. As a consequence, there exists a surjection $\mathrm{f}$ from $X(E')$ to $X(L(E'))$ such that for any expression $G$ in $X(E')$, $\mathrm{f}(G)=L(G)$ belongs to $X(L(E'))$. As a consequence:
  
  \centerline{$a^{-1}(L(E))=
      \left\{
        \begin{array}{l@{\ }l}
          & \bigcup_{E''\in X(E'),b\in\Sigma}( \mathbb{F}_{k-\mathcal{F}((a),(b))}(b^{-1}(L(E''))))\\
          \cup &   \bigcup_{E''\in X(E')} \mathbb{F}_{k-\mathcal{F}((a),(\varepsilon))}(L(E''))\\
          \cup & a^{-1}(\bigcup_{E''\in X(E'),b \in\Sigma,\mathcal{F}((\varepsilon),(b))\neq 0}( \mathbb{F}_{k-\mathcal{F}((\varepsilon),(b))}(b^{-1}(L(E'')))))\\
        \end{array}
      \right.
    $}
    
  By induction over $E'$, for any derivative $E''$ of $E'$, $b^{-1}(L(E''))=L(\frac{d'}{d'_b}(E''))$. Consequently:
  
  \centerline{
   \begin{tabular}{l@{\ }l}
    $a^{-1}(L(E))$ & $=
      \left\{
        \begin{array}{l@{\ }l}
          & \bigcup_{E''\in X(E'),b\in\Sigma}( \mathbb{F}_{k-\mathcal{F}((a),(b))}(L(\frac{d'}{d'_b}(E''))))\\
          \cup &   \bigcup_{E''\in X(E')} \mathbb{F}_{k-\mathcal{F}((a),(\varepsilon))}(L(E''))\\
          \cup & a^{-1}(\bigcup_{E''\in X(E'),b \in\Sigma,\mathcal{F}((\varepsilon),(b))\neq 0}( \mathbb{F}_{k-\mathcal{F}((\varepsilon),(b))}(L(\frac{d'}{d'_b}(E''))))\\
        \end{array}
      \right.$\\
      \\
    & $=
      \left\{
        \begin{array}{l@{\ }l}
          & L(\sum_{E''\in X(E'),b\in\Sigma}( \mathbb{F}_{k-\mathcal{F}((a),(b))}(\frac{d'}{d'_b}(E''))))\\
          \cup &   L(\sum_{E''\in X(E')} \mathbb{F}_{k-\mathcal{F}((a),(\varepsilon))}(E''))\\
          \cup & a^{-1}(L(\sum_{E''\in X(E'),b \in\Sigma,\mathcal{F}((\varepsilon),(b))\neq 0}( \mathbb{F}_{k-\mathcal{F}((\varepsilon),(b))}(\frac{d'}{d'_b}(E''))))\\
        \end{array}
      \right.
    $
   \end{tabular}
  }
    
   Furthermore, by recurrence over $k$, for any $\mathcal{F}((\varepsilon),(b))>0$, it holds:
   
   \centerline{ $a^{-1}(L(\mathbb{F}_{k-\mathcal{F}((\varepsilon),(b))}(\frac{d'}{d'_b}(E''))))=L(\frac{d'}{d'_a}(\mathbb{F}_{k-\mathcal{F}((\varepsilon),(b))}(\frac{d'}{d'_b}(E''))))$.}
   
   Finally,
   
   \centerline{
    \begin{tabular}{l@{\ }l}
     $a^{-1}(L(E))$ & $=
      \left\{
        \begin{array}{l@{\ }l}
          & L(\sum_{E''\in X(E'),b\in\Sigma}( \mathbb{F}_{k-\mathcal{F}((a),(b))}(\frac{d'}{d'_b}(E''))))\\
          \cup &   L(\sum_{E''\in X(E')} \mathbb{F}_{k-\mathcal{F}((a),(\varepsilon))}(E''))\\
          \cup & L(\frac{d'}{d'_a}(\sum_{E''\in X(E'),b \in\Sigma,\mathcal{F}((\varepsilon),(b))\neq 0}( \mathbb{F}_{k-\mathcal{F}((\varepsilon),(b))}(\frac{d'}{d'_b}(E''))))\\
        \end{array}
      \right.
    $\\
    & $=L(\frac{d'}{d'_a}(E))$.
   \end{tabular}
  }
    
    %\cqfd
  \end{proof}
  
  \begin{lemma}\label{lem eps are}
    Let $E=\mathbb{F}_{k}(E')$ be an ARE over an alphabet $\Sigma$ and $a$ be a symbol in $\Sigma$. Let $\mathcal{W}_\mathcal{F}$ and $X(E')$ be the sets defined by:
    
    \centerline{$\mathcal{W}_\mathcal{F}=(\bigcup_{b\in\Sigma,\mathcal{F}((\varepsilon),(b))=0} \{b\})^*$}
    
    \centerline{and $X(E')=\{E'\}\cup\bigcup_{w\in\mathcal{W}_\mathcal{F}}\frac{d'}{d'_w}(E')$.}
    
    Let us consider the language $L'$ defined by:
    
    \centerline{$L'=\bigcup_{F\in X(E')}L(F)\cup L(\sum_{F\in X(E'),b \in\Sigma,\mathcal{F}((\varepsilon),(b))\neq 0}( \mathbb{F}_{k-\mathcal{F}((\varepsilon),(b))}(\frac{d'}{d'_b}(F))))$.}
    
    Then the two following propositions are equivalent:
    
    \begin{itemize}
    \item $\varepsilon\in L(E)$
    \item $k\neq\indef\ \wedge\ \varepsilon\in L'$.
  \end{itemize}
    
    Furthermore, this equivalence defines a membership test that halts.
  \end{lemma}
  \begin{proof}
  
    Let $\mathcal{W}_\mathcal{F}=(\bigcup_{b\in\Sigma,\mathcal{F}((\varepsilon),(b))=0} \{b\})^*$, $X(E')=\{E'\}\cup\bigcup_{w\in\mathcal{W}_\mathcal{F}}\frac{d'}{d'_w}(E')$ and 
    for any symbol $\alpha,\beta$ in $\Sigma$, let us set $k_{\alpha,\beta}=k-\mathcal{F}((\alpha),(\beta))$. Obviously, $k=\indef\Rightarrow \varepsilon\notin L(E)$. Consequently, if $k\neq\indef$:
    
    \noindent$\varepsilon\in L(E)$ $\Leftrightarrow$ $\exists w\in L(E'), \mathbb{F}(\varepsilon,w)\in\Intervalle{0}{k}$
    
    \noindent $\Leftrightarrow$
      $\left\{
        \begin{array}{l@{\ }l}
          & \exists w\in L(E'), \mathbb{F}(\varepsilon,w)=0\\
          \vee & \exists b\in\Sigma,\exists w_1bw_2\in L(E'),\\
          & \mathbb{F}(\varepsilon,w_1)=0\wedge \mathcal{F}((\varepsilon),(b))\neq 0 \wedge \mathbb{F}(\varepsilon,w_2)\leq k_{\varepsilon,b}\\
        \end{array}
      \right.$  
      
      \vspace{\baselineskip}  
    
    \noindent $\Leftrightarrow$ $
      \left\{
        \begin{array}{l@{\ }l}
          & \exists w\in L(E'), w\in \mathcal{W}_\mathcal{F}\\
          \vee & \exists b\in\Sigma, \exists w_1\in \mathcal{W}_\mathcal{F},\exists w_2\in (w_1b)^{-1}(L(E')),\\
          & \mathcal{F}((\varepsilon),(b))\neq 0 \wedge \mathbb{F}(\varepsilon,w_2)\leq k_{\varepsilon,b}\\
        \end{array}
      \right.$   
      
      \vspace{\baselineskip} 
    
    \noindent $\Leftrightarrow$ $
      \left\{
        \begin{array}{l@{\ }l}
          & \exists w\in L(E'), \varepsilon\in w^{-1}(L(F))\\
          \vee & \exists b\in\Sigma, \exists w_2\in (b)^{-1}(\bigcup_{F\in X(E')}L(F)), \mathcal{F}((\varepsilon),(b))\neq 0 \wedge \mathbb{F}(\varepsilon,w_2)\leq k_{\varepsilon,b}\\
        \end{array}
      \right.$     
      
      \vspace{\baselineskip}
    
    \noindent $\Leftrightarrow$ $
      \left\{
        \begin{array}{l@{\ }l}
          & \varepsilon\in \bigcup_{F\in X(E')}L(F)\\
          \vee & \exists b\in\Sigma, \exists w_2\in L(\sum_{F\in X(E')} \frac{d'}{d'_b}(F)), \mathcal{F}((\varepsilon),(b))\neq 0 \wedge \mathbb{F}(\varepsilon,w_2)\leq k_{\varepsilon,b}\\
        \end{array}
      \right.$   
      
      \vspace{\baselineskip}     
    
    \noindent $\Leftrightarrow$ $
            \varepsilon\in \bigcup_{F\in X(E')}L(F) \vee \varepsilon \in L(\sum_{b\in\Sigma,F\in X(E'),\mathcal{F}((\varepsilon),(b))\neq 0} \mathbb{F}_{k_{\varepsilon,b}}(\frac{d'}{d'_b}(F))) $    
    
    Furthermore, \textbf{(a)} by induction over $E'$, the membership test
    defined by
    $\varepsilon\in \bigcup_{F\in X(E')}L(F)$ halts; \textbf{(b)} by recurrence over $k$ since $k_{\varepsilon,b}<k$ when $\mathcal{F}((\varepsilon),(b))\neq 0$, the membership test defined by:
    
    \centerline{$\varepsilon\in L(\sum_{F\in X(E'),b \in\Sigma,\mathcal{F}((\varepsilon),(b))\neq 0}( \mathbb{F}_{k-\mathcal{F}((\varepsilon),(b))}(\frac{d'}{d'_b}(F))))$}
    
    halts.    
    %\cqfd
  \end{proof}
  
  Lemma~\ref{lem deriv ham fini} ensures that the derivative automaton $B'(E)$ of an ARE $E$, computed from the set $\mathcal{D}_E$ of dissimilar derivatives of $E$ following the classical way, is a finite recognizer. Lemma~\ref{lem eps are} ensures that the set of final states can be computed, since the number of derivatives is finite. Finally, Lemma~\ref{lem form deriv ham ok} ensures that the DFA $D$ recognizes $L(E)$.
  
\begin{definition}
  Let $E$ be an ARE over an alphabet $\Sigma$. The tuple $B'(E)=(\Sigma,Q,I,F,\delta)$ is defined by:
  \begin{itemize}
    \item $Q=\mathcal{D}_E$,
    \item $I=\{E\}$,
    \item $F=\{q\in Q\mid\mathrm{r}(\varepsilon,q)=1\}$,
    \item $\forall (q,a)\in Q\times \Sigma$, $\delta(q,a)=\{\frac{d'}{d'_a}(q)\}$.
  \end{itemize}
\end{definition}
  
  \begin{proposition}
    Let $E$ be an approximate regular expression. Then:
    
    \centerline{
      $B'(E)$ is a deterministic automaton that recognizes $L(E)$.
    }
  \end{proposition}
  \begin{proof}
    Let $B'(E)=(\Sigma,Q,I,F,\delta)$. Let $w$ be a word in $\Sigma^*$. Let us show by recurrence over the length of $w$ that $\delta(E,w)=\{\frac{d'}{d'_w}(E)\}$.
    
    If $w\in\Sigma$, proposition is satisfied by definition of $\delta$.
    
    If $w=w'a$ with $w'\in\Sigma^*$ and $a\in\Sigma$, by recurrence hypothesis it holds $\delta(E,w')=\{\frac{d'}{d'_{w'}}(E)\}$. By definition of $\delta$:
    
    \centerline{
      \begin{tabular}{l@{\ }l}
        $\delta(E,w'a)$ & $=\delta(\delta(E,w'),a)$,\\
        & $=\delta(\{\frac{d'}{d'_{w'}}(E)\},a)$\\ 
        & $=\delta(\frac{d'}{d'_{w'}}(E),a)$\\ 
        & $=\{\frac{d'}{d'_a}(\frac{d'}{d'_{w'}}(E))\}$\\ 
        & $=\{\frac{d'}{d'_{w'}}(E)\}$.
      \end{tabular}
    } 

    As a first consequence, since $\mathrm{Card}(I)=1$,  since $\delta$ is a function from $Q\times\Sigma^*$ to $2^Q$, and since for any pair $(q,a)$ in $Q\times \Sigma$, $\mathrm{Card}(\delta(q,a))=1$, then the tuple $B'(E)$ is a deterministic automaton. Moreover,
    
    \centerline{
      \begin{tabular}{l@{\ }l}
        $w\in L(B'(E))$ & $\Leftrightarrow \delta(E,w'a)\cap F\neq\emptyset$\\
        & $\Leftrightarrow \{\frac{d'}{d'_{w'}}(E)\}\cap F\neq\emptyset$\\ 
        & $\Leftrightarrow \frac{d'}{d'_{w'}}(E)\in F$\\ 
        & $\Leftrightarrow \mathrm{r}(\varepsilon,\frac{d'}{d'_{w'}}(E))=1$\\ 
        & $\Leftrightarrow \varepsilon\in L(\frac{d'}{d'_{w'}}(E))$\\
        & $\Leftrightarrow \varepsilon\in w^{-1}(L(E))$\\ 
        & $\Leftrightarrow w\in L(E)$\\ 
      \end{tabular}
    } 
    
    %\cqfd
  \end{proof}
  
  For any ARE $E$, the automaton $B'(E)$ is the \emph{dissimilar derivative finite automaton} of $E$. Consequently, according to Kleene theorem, we have the following corollary.
  
  \begin{corollary}\label{cor lang are rat}
    The language denoted by any ARE is regular.
  \end{corollary}
  
  \subsection{Antimirov Derivatives for an ARE}\label{subsec ant der}
  
  Partial derivatives are defined by means of sets of expressions instead of
expressions and thus lead to the construction of a nondeterministic
recognizer. We now extend partial derivatives to the family of AREs.  
  For convenience, let us set for $\mathcal{E}$ a set of expressions $\mathbb{F}_{k}(\mathcal{E})=\bigcup_{E\in\mathcal{E}} \mathbb{F}_{k}(E)$ and $L(\mathcal{E})=\bigcup_{E\in\mathcal{E}} L(E)$.
  
  \begin{definition}\label{def part deriv are}
    Let $E=\mathbb{F}_{k}(E')$ be an ARE over an alphabet $\Sigma$ where $\mathbb{F}$ is associated with $\mathcal{F}$ and $a$ be a symbol in $\Sigma$. The \emph{partial derivative} of $E$ w.r.t. $a$ is the set $\frac{\partial}{\partial_a}(E)$ computed as follows:
    
    \centerline{$\frac{\partial}{\partial_a}(E)=
      \left\{
        \begin{array}{l@{\ }l}
          & \bigcup_{F\in X(E'),b\in\Sigma}( \mathbb{F}_{k-\mathcal{F}((a),(b))}(\frac{\partial}{\partial_b}(F)))
          \cup \    \bigcup_{F\in X(E')} \mathbb{F}_{k-\mathcal{F}((a),(\varepsilon))}(F)\\
          \cup & \frac{\partial}{\partial_a}(\bigcup_{F\in X(E'),b \in\Sigma,\mathcal{F}((\varepsilon),(b))\neq 0}( \mathbb{F}_{k-\mathcal{F}((\varepsilon),(b))}(\frac{\partial}{\partial_b}(F))))\\
        \end{array}
      \right.
    $,}
  
  \centerline{
    where $\mathcal{W}_\mathcal{F}=(\bigcup_{b\in\Sigma,\mathcal{F}((\varepsilon),(b))=0} \{b\})^*$ and $X(E')=\{E'\}\cup\bigcup_{w\in\mathcal{W}_\mathcal{F}}\frac{\partial}{\partial_w}(E')$.
  }
  \end{definition}
  
  \begin{lemma}\label{lem form deriv part ham ok}
    Let $E=\mathbb{F}_{k}(E')$ be an ARE over an alphabet $\Sigma$ and $a$ be a symbol in $\Sigma$. Then $ L(\frac{\partial}{\partial_a}(E))=a^{-1}(L(E))$.
  \end{lemma}
  \begin{proof} 
    
    By induction over the structure of $E$.
    
    According to Lemma~\ref{lem quot app lang}:
    
    \centerline{$a^{-1}(L(E))=
      \left\{
        \begin{array}{l@{\ }l}
          & \bigcup_{L''\in X(L(E')),b\in\Sigma}( \mathbb{F}_{k-\mathcal{F}((a),(b))}(b^{-1}(L'')))\\
          \cup &   \bigcup_{L''\in X(L(E'))} \mathbb{F}_{k-\mathcal{F}((a),(\varepsilon))}(L'')\\
          \cup & a^{-1}(\bigcup_{L''\in X(L(E')),b \in\Sigma,\mathcal{F}((\varepsilon),(b))\neq 0}( \mathbb{F}_{k-\mathcal{F}((\varepsilon),(b))}(b^{-1}(L''))))\\
        \end{array}
      \right.
    $,}
  
  \centerline{
    where $X(L(E'))=\{L(E')\}\cup\bigcup_{w\in\mathcal{W}_\mathcal{F}}w^{-1}(L(E'))$ with $\mathcal{W}_\mathcal{F}=(\bigcup_{b\in\Sigma,\mathcal{F}((\varepsilon),(b))=0} \{b\})^*$.
  }
  
  Let $X(E')=\{E'\}\cup\bigcup_{w\in\mathcal{W}_\mathcal{F}}\frac{\partial}{\partial_w}(E')$.
  
  By induction over $E'$, for any word $w$ in $\Sigma^*$, $w^{-1}(L(E'))=L(\frac{\partial}{d'_w}(E'))$. As a consequence:
  
  \centerline{$\bigcup_{L''\in X(L(E'))}L''= \bigcup_{E''\in X(E')} L(E'')$}
  
  and:
  
  \centerline{$a^{-1}(L(E))=
      \left\{
        \begin{array}{l@{\ }l}
          & \bigcup_{E''\in X(E'),b\in\Sigma}( \mathbb{F}_{k-\mathcal{F}((a),(b))}(b^{-1}(L(E''))))\\
          \cup &   \bigcup_{E''\in X(E')} \mathbb{F}_{k-\mathcal{F}((a),(\varepsilon))}(L(E''))\\
          \cup & a^{-1}(\bigcup_{E''\in X(E'),b \in\Sigma,\mathcal{F}((\varepsilon),(b))\neq 0}( \mathbb{F}_{k-\mathcal{F}((\varepsilon),(b))}(b^{-1}(L(E'')))))\\
        \end{array}
      \right.
    $}
    
  By induction over $E'$, for any derivative $E''$ of $E'$, it holds 
  
  \centerline{$b^{-1}(L(E''))=L(\frac{\partial}{\partial_b}(E''))$.}
  
  Consequently:
  
  \centerline{
    \begin{tabular}{l@{\ }l}
     $a^{-1}(L(E))$ & $=
      \left\{
        \begin{array}{l@{\ }l}
          & \bigcup_{E''\in X(E'),b\in\Sigma}( \mathbb{F}_{k-\mathcal{F}((a),(b))}(L(\frac{\partial}{\partial_b}(E''))))\\
          \cup &   \bigcup_{E''\in X(E')} \mathbb{F}_{k-\mathcal{F}((a),(\varepsilon))}(L(E''))\\
          \cup & a^{-1}(\bigcup_{E''\in X(E'),b \in\Sigma,\mathcal{F}((\varepsilon),(b))\neq 0}( \mathbb{F}_{k-\mathcal{F}((\varepsilon),(b))}(L(\frac{\partial}{\partial_b}(E''))))\\
        \end{array}
      \right.
    $\\
    \\
    & $=
      \left\{
        \begin{array}{l@{\ }l}
          & L(\bigcup_{E''\in X(E'),b\in\Sigma}( \mathbb{F}_{k-\mathcal{F}((a),(b))}(\frac{\partial}{\partial_b}(E''))))\\
          \cup &   L(\bigcup_{E''\in X(E')} \mathbb{F}_{k-\mathcal{F}((a),(\varepsilon))}(E''))\\
          \cup & a^{-1}(L(\bigcup_{E''\in X(E'),b \in\Sigma,\mathcal{F}((\varepsilon),(b))\neq 0}( \mathbb{F}_{k-\mathcal{F}((\varepsilon),(b))}(\frac{\partial}{\partial_b}(E''))))\\
        \end{array}
      \right.
    $\\
   \end{tabular}
   }
    
   Furthermore, by recurrence over $k$, for any $\mathcal{F}((\varepsilon),(b))>0$, it holds:
   
   \centerline{ $a^{-1}(L(\mathbb{F}_{k-\mathcal{F}((\varepsilon),(b))}(\frac{\partial}{\partial_b}(E''))))=L(\frac{\partial}{\partial_a}(\mathbb{F}_{k-\mathcal{F}((\varepsilon),(b))}(\frac{\partial}{\partial_b}(E''))))$.}
   
   Finally,
   
   \centerline{$a^{-1}(L(E))=
      \left\{
        \begin{array}{l@{\ }l}
          & L(\bigcup_{E''\in X(E'),b\in\Sigma}( \mathbb{F}_{k-\mathcal{F}((a),(b))}(\frac{\partial}{\partial_b}(E''))))\\
          \cup &   L(\bigcup_{E''\in X(E')} \mathbb{F}_{k-\mathcal{F}((a),(\varepsilon))}(E''))\\
          \cup & L(\frac{\partial}{\partial_a}(\bigcup_{E''\in X(E'),b \in\Sigma,\mathcal{F}((\varepsilon),(b))\neq 0}( \mathbb{F}_{k-\mathcal{F}((\varepsilon),(b))}(\frac{\partial}{\partial_b}(E''))))\\
        \end{array}
      \right.
    $}
    
   and
   
   \centerline{$a^{-1}(L(E))=L(\frac{\partial}{\partial_a}(E))$.}
    
    %\cqfd 
  \end{proof}
  
  Let $\mathcal{DT}_E$ be the set of derivated terms of an ARE $E$, that is the set
of the elements of all the partial derivatives of $E$.
  
  \begin{lemma}\label{lem deriv part fini}
    Let $E=\mathbb{F}_{k}(E')$ be an ARE over an alphabet $\Sigma$. Then:
    
    \centerline{$\mathcal{DT}_E \subset \bigcup_{k'\in\{0,\ldots,k\}} \mathbb{F}_{k'}(\mathcal{DT}_{E'})$.}
    
    \noindent Moreover, the computation of $\mathcal{DT}_E$ halts.
  \end{lemma}  
  \begin{proof}
    Consider that $\mathbb{F}$ is associated with $\mathcal{F}$. Let us define the set $S(E',k)=\bigcup_{k'\in\Intervalle{0}{k}}\mathbb{F}_{k'}(\mathcal{DT}_{E'})$. Let us show by induction over the structure of $E'$ and by recurrence over $k$ that $\mathcal{DT}_E\subset S(E',k)$.  
  Since $X(E')$ is a finite set of derivated terms of $E'$, any subexpression of type $\mathbb{F}_{k-\mathcal{F}((a),(\varepsilon))}(F)$ with $F\in X(E')$ belongs to $S(E',k)$. Since $X(E')$ is a subset of $\mathcal{DT}_{E'}$, $\frac{\partial}{\partial_b}(F)$ is a set of derivated terms of $E'$ for any $b$ in $\Sigma$. Consequently, $\bigcup_{F\in X(E'),b\in\Sigma}( \mathbb{F}_{k-\mathcal{F}((a),(b))}(\frac{\partial}{\partial_b}(F)))$ is a subset of $S(E',k)$.
  Finally, by recurrence hypothesis, for $k'< k$, any partial derivative of an expression $\mathbb{F}_{k'}(H)$ is a subset of $S(H,k')$. Consequently, any partial derivative of $\mathbb{F}_{k-\mathcal{F}((\varepsilon),(b))}(\frac{\partial}{\partial_b}(F))$ is included into $\bigcup_{F'\in \frac{\partial}{\partial_b}(F)} S(F',k-\mathcal{F}((\varepsilon),(b)))$  if $\mathcal{F}((\varepsilon),(b))\neq 0$. Since $F$ is a derivated term of $E'$, so is any expression in $\frac{\partial}{\partial_b}(F)$, and since $k-\mathcal{F}((\varepsilon),(b))\leq k$, any partial derivative of $\mathbb{F}_{k-\mathcal{F}((\varepsilon),(b))}(\frac{\partial}{\partial_b}(F))$ is a subset of $S(E',k)$. As a consequence, \textbf{(Fact~A)} any derivated term of $E$ w.r.t. a symbol $a$ belongs to $S(E',k)$.
  
  Furthermore, let us show that if $G=\mathbb{F}_{k'}(H)$ is an expression that belongs to $S(E',k)$, then any partial derivative of $G$ is a subset of $S(E',k)$.
  According to \textbf{Fact~A}, any partial derivative of an expression $\mathbb{F}_{k'}(H)$ is a subset of $S(H,k')$. When $H$ is a derivated term of $E'$ and $k'\leq k$, any expression in $S(H,k')$ belongs to $S(E',k)$. As a consequence, any derivated term of $E$ belongs to $S(E',k)$.
  
  As a conclusion, $\mathcal{DT}_E\subset S(E',k)=\bigcup_{k'\in\Intervalle{0}{k}}\mathbb{F}_{k'}(\mathcal{DT}_{E'})$. Moreover, by induction over $E'$ and by recurrence over $k$, since any derivated term of an expression $F$ in $X(E')$ belongs to the finite set of derivated terms of $E'$ the computation of which halts, and since $k-\mathcal{F}((\varepsilon),(b))<k$ when $\mathcal{F}((\varepsilon),(b))\neq 0$, the computation of $\mathcal{DT}_E$ halts.  
    %\cqfd
  \end{proof}
  
  \begin{corollary}\label{cor card terme mot}
    Let $E=\mathbb{F}_{k}(E')$ be an ARE over an alphabet $\Sigma$. Then $\mathcal{DT}_E$ is a finite set of AREs. Furthermore, $\mathrm{Card}(\mathcal{DT}_E) \leq \mathrm{Card}(\mathcal{DT}_{E'})\times (k+1) $.
  \end{corollary}
  
    \begin{lemma}\label{lem eps deriv part are}
    Let $E=\mathbb{F}_{k}(E')$ be an ARE over an alphabet $\Sigma$ and $a$ be a symbol in $\Sigma$. Let $\mathcal{W}_\mathcal{F}$ and $X(E')$ be the sets defined by:
    
    \centerline{  $\mathcal{W}_\mathcal{F}=(\bigcup_{b\in\Sigma,\mathcal{F}((\varepsilon),(b))=0} \{b\})^*$}
    
    \centerline{ and $X(E')=\{E'\}\cup\bigcup_{w\in\mathcal{W}_\mathcal{F}}\frac{\partial}{\partial_w}(E')$.}
    
    Let $L'$ be the language defined by:
    
    \centerline{$L'=\bigcup_{F\in X(E')}L(F)\cup L(\bigcup_{F\in X(E'),b \in\Sigma,\mathcal{F}((\varepsilon),(b))\neq 0}( \mathbb{F}_{k-\mathcal{F}((\varepsilon),(b))}(\frac{\partial}{\partial_b}(F))))$.}
    
     Then the two following conditions are equivalent:
    
    \begin{itemize}
    \item $\varepsilon\in L(E)$,
    
    \item $k\neq\indef \wedge \varepsilon\in L'$.
    \end{itemize}
    
    \noindent Furthermore, this equivalence defines a membership test that halts.
  \end{lemma}    
  \begin{proof}
  Let $\mathcal{W}_\mathcal{F}=(\bigcup_{b\in\Sigma,\mathcal{F}((\varepsilon),(b))=0} \{b\})^*$ and $X(E')=\{E'\}\cup\bigcup_{w\in\mathcal{W}_\mathcal{F}}\frac{\partial}{\partial_w}(E')$. For convenience, for any two symbols $\alpha$ and $\beta$ in $\Sigma\cup\{\varepsilon\}$, let us set $k_{\alpha,\beta}=k-\mathcal{F}((\alpha),(\beta))$. Obviously, if $k=\indef$, $\varepsilon\notin L(E)$. For $k\neq\indef$:
  
    $\varepsilon\in L(E)$ $\Leftrightarrow$ $\exists w\in L(E'), \mathbb{F}(\varepsilon,w)\in\Intervalle{0}{k}$
    
    \noindent $\Leftrightarrow$ $
      \left\{
        \begin{array}{l@{\ }l}
          & \exists w\in L(E'), \mathbb{F}(\varepsilon,w)=0\\
          \vee & \exists b\in\Sigma,\exists w_1bw_2\in L(E'),\\
          &  \mathbb{F}(\varepsilon,w_1)=0\wedge \mathcal{F}((\varepsilon),(b))\neq 0 \wedge \mathbb{F}(\varepsilon,w_2)\leq k_{(\varepsilon),(b))}\\
        \end{array}
      \right.$   
      
      \vspace{\baselineskip} 
    
    \noindent $\Leftrightarrow$ $
      \left\{
        \begin{array}{l@{\ }l}
          & \exists w\in L(E'), w\in \mathcal{W}_\mathcal{F}\\
          \vee & \exists b\in\Sigma, \exists w_1\in \mathcal{W}_\mathcal{F},\exists w_2\in (w_1b)^{-1}(L(E')),\\
          & \mathcal{F}((\varepsilon),(b))\neq 0 \wedge \mathbb{F}(\varepsilon,w_2)\leq k_{(\varepsilon),(b))}\\
        \end{array}
      \right.$   
      
      \vspace{\baselineskip} 
    
    \noindent $\Leftrightarrow$ $
      \left\{
        \begin{array}{l@{\ }l}
          & \exists w\in L(E'), \varepsilon\in w^{-1}(\mathcal{W}_\mathcal{F})\\
          \vee & \exists b\in\Sigma, \exists w_2\in (b)^{-1}(\bigcup_{F\in X(E')}L(F)),\\
          & \mathcal{F}((\varepsilon),(b))\neq 0 \wedge \mathbb{F}(\varepsilon,w_2)\leq k_{(\varepsilon),(b))}\\
        \end{array}
      \right.$  
      
      \vspace{\baselineskip}   
    
    \noindent $\Leftrightarrow$ $
      \left\{
        \begin{array}{l@{\ }l}
          & \varepsilon\in \bigcup_{F\in X(E')}L(F)\\
          \vee & \exists b\in\Sigma, \exists w_2\in L(\bigcup_{F\in X(E')} \frac{\partial}{\partial_b}(F)),\\
          & \mathcal{F}((\varepsilon),(b))\neq 0 \wedge \mathbb{F}(\varepsilon,w_2)\leq k_{(\varepsilon),(b))}\\
        \end{array}
      \right.$  
      
      \vspace{\baselineskip}      
    
    \noindent $\Leftrightarrow$ $
      \left\{
        \begin{array}{l@{\ }l}
          & \varepsilon\in \bigcup_{F\in X(E')}L(F)\\
          \vee & \varepsilon \in L(\bigcup_{b\in\Sigma,F\in X(E'),\mathcal{F}((\varepsilon),(b))\neq 0} \mathbb{F}_{k_{(\varepsilon),(b))}}(\frac{\partial}{\partial_b}(F)))\\
        \end{array}
      \right.$

    Furthermore, \textbf{(a)} by induction over $E'$, the membership test
    defined by
    $\varepsilon\in \bigcup_{F\in X(E')}L(F)$ halts; \textbf{(b)} by recurrence over $k$ since 
    $k_{\varepsilon),b}<k$
    when $\mathcal{F}((\varepsilon),(b))\neq 0$, the membership test defined by
    
    \centerline{$\varepsilon\in L(\bigcup_{F\in X(E'),b \in\Sigma,\mathcal{F}((\varepsilon),(b))\neq 0}( \mathbb{F}_{k-\mathcal{F}((\varepsilon),(b))}(\frac{\partial}{\partial_b}(F))))$}
    
     halts. 
    %\cqfd
  \end{proof}
  
  Corollary~\ref{cor card terme mot} ensures that the derivated term automaton $A(E)$ of an ARE $E$, computed from the set $\mathcal{DT}_E$ of derivated terms of $E$ following the classical way, is a finite recognizer. Lemma~\ref{lem eps deriv part are} ensures that the set of final states can be computed. Finally, Lemma~\ref{lem form deriv part ham ok} ensures that the NFA $A$ recognizes $L(E)$.

\begin{definition}
  Let $E$ be an ARE over an alphabet $\Sigma$. The tuple $A(E)=(\Sigma,Q,I,F,\delta)$ is defined by:
  \begin{itemize}
    \item $Q=\mathcal{DT}_E$,
    \item $I=\{E\}$,
    \item $F=\{q\in Q\mid\mathrm{r}(\varepsilon,q)=1\}$,
    \item $\forall (q,a)\in Q\times \Sigma$, $\delta(q,a)=\frac{\partial}{\partial_a}(q)$.
  \end{itemize}
\end{definition}
  
  \begin{proposition}
    Let $E$ be an approximate regular expression. Then:
    
    \centerline{
      $A(E)$ is a finite automaton that recognizes $L(E)$.
    }
  \end{proposition}
  \begin{proof}
    Let $A(E)=(\Sigma,Q,I,F,\delta)$. Let $w$ be a word in $\Sigma^*$. Let us show by recurrence over the length of $w$ that $\delta(E,w)=\frac{\partial}{\partial_w}(E)$.
    
    If $w\in\Sigma$, proposition is satisfied by definition of $\delta$.
    
    If $w=w'a$ with $w'\in\Sigma^*$ and $a\in\Sigma$, by recurrence hypothesis it holds $\delta(E,w')=\frac{\partial}{\partial_{w'}}(E)$. By definition of $\delta$:
    
    \centerline{
      \begin{tabular}{l@{\ }l}
        $\delta(E,w'a)$ & $=\delta(\delta(E,w'),a)$,\\
        & $=\delta(\frac{\partial}{\partial_{w'}}(E),a)$\\ 
        & $=\bigcup_{E'\in \frac{\partial}{\partial_{w'}}(E)}\delta(E',a)$\\ 
        & $=\bigcup_{E'\in \frac{\partial}{\partial_{w'}}(E)}\frac{\partial}{\partial_{a}}(E')$\\
        & $=\frac{\partial}{\partial_{w'a}}(E)$\\
      \end{tabular}
    } 

    Consequently,
    
    \centerline{
      \begin{tabular}{l@{\ }l}
        $w\in L(A(E))$ & $\Leftrightarrow \delta(E,w'a)\cap F\neq\emptyset$\\
        & $\Leftrightarrow \frac{\partial}{\partial_{w'a}}(E)\cap F\neq\emptyset$\\ 
        & $\Leftrightarrow \exists E'\in \frac{\partial}{\partial_{w'a}}(E)\mid E' \in F$\\
        & $\Leftrightarrow \exists E'\in \frac{\partial}{\partial_{w'a}}(E)\mid \mathrm{r}(\varepsilon,E')=1$\\ 
        & $\Leftrightarrow \varepsilon\in \bigcup_{E'\in \frac{\partial}{\partial_{w'a}}(E)} L(E')$\\ 
        & $\Leftrightarrow \varepsilon\in w^{-1}(L(E))$\\ 
        & $\Leftrightarrow w\in L(E)$\\ 
      \end{tabular}
    } 
    
    %\cqfd
  \end{proof}
  
  For any ARE $E$, the automaton $A(E)$ is called the \emph{derivated term finite automaton} of $E$.
  
  \subsection{Back to Hamming and Levenshtein Derivation}\label{se:lienEntreDeuxForm}
  
  This subsection is devoted to show the link between HLARE derivation formulae and ARE ones. Given an HLARE $E$ and a word $w$, the following proposition illustrates the fact that the expression $D'_w(E)$ of Definition~\ref{def deriv diss hlare} (resp. the set of expressions $\Delta_w(E)$ of Definition~\ref{def deriv part hlare}) and the expression $\frac{d'}{d'_w}(E)$ in Definition~\ref{def diss deriv are} (resp. the set of expressions $\frac{\partial}{\partial_w}(E)$ in Definition~\ref{def part deriv are}) are syntactically equal up to the expression $\emptyset$.
  
  \begin{proposition}\label{prop eq der HLARE ARE}
    Let $E$ be an HLARE over an alphabet $\Sigma$. For any symbol $a$ in $\Sigma$, the two following conditions are satisfied:
    
    \begin{itemize}
      \item $\frac{d'}{d'_a}(E)\in\{(D'_a(E)+\emptyset)_{\sim_s},D'_a(E)\}$,
      \item $\frac{\partial}{\partial_a}(E)\in\{\Delta_a(E)\cup\{\emptyset\},\Delta_a(E)\}$.
    \end{itemize}
  \end{proposition}
  \begin{proof}
    We prove the first membership relation. A similar proof can be given for the second one. 
    
    By induction over the structure of an HLARE. 
    
    \begin{enumerate}
    
    \item If $E=a\in\Sigma$, $E=E_1+E_2$, $E=E_1\cdot E_2$ or if $E=E_1^*$, the proposition is trivially checked by similarity of the formulae.
    
    \item If $E=\mathbb{H}_k(E')$, by definition of $\frac{d'}{d'_a}(E)$:
    
    \centerline{$\frac{d'}{d'_a}(E)=
      \left(
        \begin{array}{l@{\ }l}
          & \sum_{F\in X(E'),b\in\Sigma}( \mathbb{H}_{k-\mathcal{H}((a),(b))}(\frac{d'}{d'_b}(F)))\\
          + &   \sum_{F\in X(E')} \mathbb{H}_{k-\mathcal{H}((a),(\varepsilon))}(F)\\
          + & \frac{d'}{d'_a}(\sum_{F\in X(E'),b \in\Sigma,\mathcal{H}((\varepsilon),(b))\neq 0}( \mathbb{H}_{k-\mathcal{H}((\varepsilon),(b))}(\frac{d'}{d'_b}(F))))\\
        \end{array}
      \right)_{\sim_s}
    $,}
  
  \centerline{
    where $X(E')=\{E'\}\cup\bigcup_{w\in\mathcal{W}_\mathcal{H}}\frac{d'}{d'_w}(E')$ with $\mathcal{W}_\mathcal{H}=(\bigcup_{b\in\Sigma,\mathcal{H}((\varepsilon),(b))=0} \{b\})^*$.
  }
  
  By definition of $\mathbb{H}$, $X(E')=\{E'\}$, $\mathcal{H}((a),(b))\in\{1,0\}$ and $\mathcal{H}((a),(\varepsilon))=\indef$ for any two symbols $a$ and $b$ in $\Sigma$.
  
  Consequently:
    
    \centerline{$\frac{d'}{d'_a}(E)=
      \left(
        \begin{array}{l@{\ }l}
          & \sum_{b\in\Sigma}( \mathbb{H}_{k-\mathcal{H}((a),(b))}(\frac{d'}{d'_b}(E')))\\
          + &   \sum_{F\in X(E')} \mathbb{H}_{k-\indef}(F)\\
          + & \frac{d'}{d'_a}(\sum_{F\in X(E'),b \in\Sigma,\mathcal{H}((\varepsilon),(b))\neq 0}( \mathbb{H}_{k-\indef}(\frac{d'}{d'_b}(F))))\\
        \end{array}
      \right)_{\sim_s}
    $,}
    
    and finally
    
    \centerline{$\frac{d'}{d'_a}(E)=
      \left(
        \begin{array}{l@{\ }l}
          & \mathbb{H}_{k}(\frac{d'}{d'_a}(E'))\\
          +& \sum_{b\in\Sigma\setminus\{a\}}( \mathbb{H}_{k-1}(\frac{d'}{d'_b}(E')))\\
          + &   \emptyset\\
          + &   \emptyset \\
        \end{array}
      \right)_{\sim_s}$}

    \centerline{$=
      \left(
        \begin{array}{l@{\ }l}
          & \mathbb{H}_{k}(D'_a(E'))\\
          +& \sum_{b\in\Sigma\setminus\{a\}}( \mathbb{H}_{k-1}(D'_b(E')))\\
          + &   \emptyset\\
          + &   \emptyset \\
        \end{array}
      \right)_{\sim_s}
    \in \{D'_a(E)+\emptyset,D'_a(E)\}$.}

    \item If $E=\mathbb{L}_k(E')$, by definition of $\frac{d'}{d'_a}(E)$:
    
    \centerline{$\frac{d'}{d'_a}(E)=
      \left(
        \begin{array}{l@{\ }l}
          & \sum_{F\in X(E'),b\in\Sigma}( \mathbb{L}_{k-\mathcal{L}((a),(b))}(\frac{d'}{d'_b}(F)))\\
          + &   \sum_{F\in X(E')} \mathbb{L}_{k-\mathcal{L}((a),(\varepsilon))}(F)\\
          + & \frac{d'}{d'_a}(\sum_{F\in X(E'),b \in\Sigma,\mathcal{L}((\varepsilon),(b))\neq 0}( \mathbb{L}_{k-\mathcal{H}((\varepsilon),(b))}(\frac{d'}{d'_b}(F))))\\
        \end{array}
      \right)_{\sim_s}
    $,}
  
  \centerline{
    where $X(E')=\{E'\}\cup\bigcup_{w\in\mathcal{W}_\mathcal{L}}\frac{d'}{d'_w}(E')$ with $\mathcal{W}_\mathcal{L}=(\bigcup_{b\in\Sigma,\mathcal{L}((\varepsilon),(b))=0} \{b\})^*$.
  }
  
  By definition of $\mathbb{L}$, $X(E')=\{E'\}$, $\mathcal{L}((a),(b))\in\{1,0\}$ and $\mathcal{L}((a),(\varepsilon))=\mathcal{L}((\varepsilon,(a)))=1$ for any two symbols $a$ and $b$ in $\Sigma$.
  
  Consequently:
    
    \centerline{$\frac{d'}{d'_a}(E)=
      \left(
        \begin{array}{l@{\ }l}
          & \sum_{b\in\Sigma}( \mathbb{L}_{k-\mathcal{L}((a),(b))}(\frac{d'}{d'_b}(E')))\\
          + &   \mathbb{L}_{k-1}(E')\\
          + & \frac{d'}{d'_a}(\sum_{b \in\Sigma}( \mathbb{L}_{k-1}(\frac{d'}{d'_b}(E'))))\\
        \end{array}
      \right)_{\sim_s}
    $}
    
    \centerline{$=
      \left(
        \begin{array}{l@{\ }l}
          & \mathbb{L}_{k}(\frac{d'}{d'_a}(E')))\\
          + & \sum_{b\in\Sigma}( \mathbb{L}_{k-1}(\frac{d'}{d'_b}(E')))\\
          + &   \mathbb{L}_{k-1}(E')\\
          + & \frac{d'}{d'_a}(\sum_{b \in\Sigma}( \mathbb{L}_{k-1}(\frac{d'}{d'_b}(E'))))\\
        \end{array}
      \right)_{\sim_s}
    $}
    
    Finally, by induction hypothesis and by recurrence over $k$,
    
    \centerline{$\frac{d'}{d'_a}(E)=
      \left(
        \begin{array}{l@{\ }l}
          & \mathbb{L}_{k}(D'_a(E')))\\
          + & \sum_{b\in\Sigma}( \mathbb{L}_{k-1}(D'_b(E')))\\
          + &   \mathbb{L}_{k-1}(E')\\
          + & D'_a(\sum_{b \in\Sigma}( \mathbb{L}_{k-1}(D'_b(E'))))\\
        \end{array}
      \right)_{\sim_s}
    =D'_a(E)$.}
    
    \end{enumerate}
    
    %\cqfd
  \end{proof}
  
  As a corollary of Proposition~\ref{prop eq der HLARE ARE}, the proofs of the lemmas and propositions of Section~\ref{se:hamLevDeriv} can be deduced from the corresponding ones of Section~\ref{se: deriv are}.
  
\section{Conclusion}

  The similarity operators
that equip the family of approximate regular expressions
 make AREs to be a nice tool to deal with
approximate regular expression matching.
The extension of dissimilar derivatives and partial derivatives to the
family of AREs allows us to provide a syntactical solution to the
approximate membership problem; moreover in each case the set of
derivatives is finite and thus this extension also yields the
construction of a recognizer.
An additional advantage of similarity operators is that they can be combined with 
other regular operators, such as intersection and complementation operators~\cite{CCM11b},
in order to produce even smaller expressions.

%\bibliography{biblio} 
%\bibliography{E:/DocsSync/Recherche/Bibliographie/biblio}
\bibliography{/home/ludo/Bureau/Donnees/DocsSync/Recherche/Bibliographie/biblio}

\end{document}